\documentclass{article}
\usepackage{array}
\usepackage{amssymb}
\usepackage{mathrsfs}
\usepackage{geometry}
\usepackage{graphicx}
\usepackage{subfigure}
\usepackage{abstract}
\usepackage{ntheorem}
\usepackage{amsmath}
\usepackage{textcomp}
\usepackage{multicol} 
\usepackage{titlesec}
\usepackage{wrapfig}
\usepackage{soul}
\usepackage{etoolbox}
\usepackage{algorithm}
\usepackage{algpseudocode}


\usepackage{tikz}

\usetikzlibrary{shapes.geometric, arrows}
\usetikzlibrary{arrows.meta, positioning, calc}
\tikzstyle{circleNode} = [circle, minimum size=0.8cm, text centered, draw=black]
\tikzstyle{arrow} = [thick,->,>=stealth]
\usepackage[hidelinks]{hyperref}
\renewcommand{\epsilon}{\varepsilon}

\makeatletter
\newcommand{\SPR}{\textit{SPR}}
\newcommand{\MSPR}{\textit{MSPR}}

\newcommand{\plat}[1]{\raisebox{0pt}[0pt][0pt]{#1}}

\newcommand{\Rmnum}[1]{\expandafter\@slowromancap\romannumeral #1@}
\makeatother
\newcommand{\powerset}{\raisebox{.15\baselineskip}{\Large\ensuremath{\wp}}}

\theorembodyfont{\upshape}
\newtheorem{theorem}{Theorem}
\newtheorem{lemma}{Lemma}

\newtheorem{observation}{Observation}
\newtheorem{example}{Example}

\newtheorem{proposition}{Proposition}
\newtheorem{corollary}{Corollary}

\newtheorem{myDef}{Definition}

\newenvironment{proof}{{\noindent\it Proof.}}{\hfill $\square$}

\newcommand{\doi}[1]{DOI: \href{https://doi.org/#1}{#1}}

\title{The Similarity Control Problem with Required Events}
\date{}

\usepackage{authblk} 

\titleformat{\section}{\normalfont\large\bfseries}{\thesection}{1em}{\large}
\author[a,d]{Yu Wang}
\author[a$^*$]{Zhaohui Zhu}
\author[d,b]{Rob van Glabbeek}
\author[c]{Jinjin Zhang}
\author[d]{Yixuan Li}

\affil[a]{College of Computer Science and Technology, Nanjing University of Aeronautics and Astronautics}
\affil[b]{School of Computer Science and Engineering, University of New South Wales}
\affil[c]{School of Computer Science, Nanjing Audit University}
\affil[d]{School of Informatics, University of Edinburgh}

\titleformat{\section}
{\normalfont\large\bfseries}{\thesection}{1em}{}

\titlespacing{\section}{0pt}{1.5ex plus 0.2ex minus 0.2ex}{1ex plus 0.2ex}

\BeforeBeginEnvironment{proof}{\vspace{-0.5\baselineskip}}
\AfterEndEnvironment{proof}{\vspace{-0.5\baselineskip}}

\begin{document}
	
	\twocolumn[	 
	\maketitle
	\vspace{-1cm}

	\begin{onecolabstract}
	In order to guarantee that a supervised system satisfies safety requirements of the specification, as well as requirements saying that in certain states certain events must be enabled, this paper introduces required events for discrete event systems and reconsiders the similarity control problem  while taking all requirements from the specification into account. The notion of a covariant-contravariant simulation,  which is  finer than the conventional notion of simulation, is adopted to act as the behavioral relation of supervisory control theory. A necessary and sufficient condition for the solvability of this problem is established and a method for synthesizing a maximally permissive supervisor  is provided.
		
		
	\end{onecolabstract}
	]

\section{Introduction}

Discrete event systems (DESs) are event-driven systems composed of discrete events which happen at distinct points in time. Supervisory control theory, originating from the work of   Ramadge and Wonham \cite{1987Supervisory}, provides a framework to explore the control of discrete event systems. In this theory, both systems---called \emph{plants}---and  specifications are  modeled as (non)deterministic automata \cite{Kimura2014Maximally,kushi2017synthesis,2019Maximally,2022Synthesis,2023MaximallyPermissive,1987Supervisory,SUN2014287,2020Synthesis,2018Nonblockingsimilarity,2007Control,2006Control}. A supervisor, acting as a controller, is also modeled as an automaton, which controls the system  to ensure that the supervised system could run as desired.

Language  equivalence  is the first criterion adopted to evaluate whether a supervised system runs as desired.
Although language equivalence is a suitable notion to capture 
behavioral equivalence for deterministic systems, it is often unsuitable in the nondeterministic case. Therefore, bisimilarity \cite{1980Robin}, which is the finest behavioral equivalence in common use \cite{GLABBEEK20013},  was introduced in supervisory control theory in \cite{2006Control}. 

In a nondeterministic setting, bisimilarity control problems are considered in \cite{takai2019bisimilarity,2006Control}, which aim to synthesize a supervisor so that the supervised system is bisimilar to  a given specification. In \cite{takai2019bisimilarity,2006Control} it is assumed that a supervisor can observe all event occurrences of the plant and the supervised system is modeled by the synchronous product of the plant and supervisor. Bisimilarity control problems are  also considered in another  framework in \cite{KatsuyukiKIMURA2014,LIU2011782}.
There the supervised plant is modeled with a relation-based control 
mechanism instead of synchronous products, and a necessary and sufficient condition for the existence of a bisimilarity enforcing supervisor is obtained in terms of simulation-based controllability \cite{KatsuyukiKIMURA2014}. Under the assumption that a supervisor can observe both event occurrences and the current state of the plant, the literature \cite{LIU2011782} also considers the bisimilarity control problem. Similar to \cite{KatsuyukiKIMURA2014}, the relation-based control mechanism is adopted to model the supervised plant, and a necessary and sufficient condition for the existence of a bisimilarity enforcing supervisor is given in terms of partial bisimilarity. The supervisor in \cite{LIU2011782} is required to meet more requirements than the one in \cite{KatsuyukiKIMURA2014}, which ensure that events enabled by a supervisor are never disabled in the supervised plant.

The ideal situation is that one can synthesize  a bisimilarity enforcing supervisor  such that the supervised system and the specification are bisimilar. However, this aim is not always achievable 	
due to the existence of nondeterministic transitions and  uncontrollable events. Hence, researchers relax the requirement of bisimilarity and consider the similarity control problem which refers to finding a supervisor  so that the supervised system is simulated by the specification \cite{2022Synthesis,2023MaximallyPermissive,2020Synthesis,2007Control}. In this situation, the behavior of the supervised system must be allowed by the specification, that is, it meets safety requirements.
In  \cite{2020Synthesis}, a necessary and sufficient condition for the solvability of the similarity control problem is obtained and a method for constructing	
maximally permissive supervisors is considered. Further, this work is extended to the nonblocking similarity control problem under partially event observations in  \cite{2023MaximallyPermissive}. 

However, similarity may be too coarse a behavioral relation, because some requirements from specifications may be neglected by the supervised systems, leading to unreasonable results in certain situations.
For example, consider an automatic check-out scanner in a shopping mall, an example taken from \cite{2006Control} with slight modifications. A plant $G$, modeled as a state machine, is represented in the left of Figure \ref{figure-motivating1}. 
\begin{figure}[hbt]
	\centering
	\begin{minipage}[t]{0.22\textwidth}
		\centering
		\begin{tikzpicture}[
			->, >=stealth,
			every node/.style={minimum size=0.5cm, align=center},
			xnode/.style={draw, circle},
			every path/.style={thick},
			font=\scriptsize
			]
			
			\node[xnode] (x0) at (0, 0)    {$x_0$};
			\node[xnode] (x1) at (0, -1.2)   {$x_1$};
			\node[xnode] (x2) at (-0.8, -2.4) {$x_2$};
			\node[xnode] (x3) at ( 0.8, -2.4) {$x_3$};
			\node[xnode] (x4) at (0, -3.6)    {$x_4$};
			
			\draw (x0) -- (x1) node[midway, right] {start};
			\draw (x1) -- (x2) node[midway, left] {scan};
			\draw (x1) -- (x3) node[midway, right] {scan};
			\draw (x2) to[bend left=15]  node[midway, right] {cancel} (x4);
			\draw (x2) to[bend right=15] node[midway, left]  {put}    (x4);
			\draw (x3) to[bend left=15]  node[midway, right] {put}    (x4);
			
			\coordinate (rightx4) at ($(x4) + (1.3, 0)$);
			\coordinate (rightx0) at ($(x0) + (1.3, 0)$);
			\draw (x4) -- (rightx4) |- node[near start, right] {pay} (x0.east);
			
			\coordinate (leftx4) at ($(x4) + (-1.3, 0)$);
			\coordinate (leftx1)  at ($(x1) + (-1.3, 0)$);
			\draw (x4) -- (leftx4) |- node[near start, left] {next} (x1.west);
		\end{tikzpicture}
	\end{minipage}
	\hfill
	\begin{minipage}[t]{0.2\textwidth}
		\centering
		\begin{tikzpicture}[
			->, >=stealth,
			every node/.style={minimum size=0.5cm, align=center},
			xnode/.style={draw, circle},
			every path/.style={thick},
			font=\scriptsize
			]

			\node[xnode] (z0) at (0, 0)    {$z_0$};
			\node[xnode] (z1) at (0, -1.2) {$z_1$};
			\node[xnode] (z2) at (0, -2.4)  {$z_2$};
			\node[xnode] (z4) at (0, -3.6)  {$z_4$};

			\draw (z0) -- (z1) node[midway, right] {start};
			\draw (z1) -- (z2) node[midway, left] {scan};
			\draw (z2) to[bend left=15]  node[midway, right] {cancel} (z4);
			\draw (z2) to[bend right=15] node[midway, left]  {put}    (z4);
	
			\coordinate (rightz4) at ($(z4) + (1.2, 0)$);
			\coordinate (rightz0) at ($(z0) + (1.2, 0)$);
			\draw (z4) -- (rightz4) |- node[near start, right] {pay} (z0.east);
			
			\coordinate (leftz4) at ($(z4) + (-1.0, 0)$);
			\coordinate (leftz1)  at ($(z1) + (-1.2, 0)$);
			\draw (z4) -- (leftz4) |- node[near start, left] {next} (z1.west);
		\end{tikzpicture}
	\end{minipage}
	\caption{The plant $G$ (left) and specification $R$ (right)}
	\label{figure-motivating1}
\end{figure}
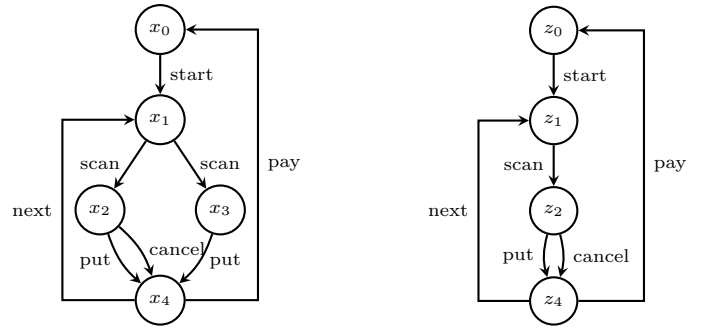

The plant $G$ starts when a customer presses the start button to initiate the check-out process. Then, an item is scanned by the scanner which nondeterministically transitions to one of two states $x_2$ and $x_3$. In the state $x_2$, the customer can choose to either put the item in a bag  or cancel the transaction, while in the state $x_3$, the only option available is to put the item in the bag. As auditors of the system we judge not providing the option to cancel in the  state $x_3$ to be unacceptable.	
Once the item is handled, the scanner waits for either a request for the next item or a request to pay whenever there are no more items. In the former case, the scanner goes back to the  state $x_1$ where the check-out process resumes, whereas it returns to the initial state $x_0$ in the latter case. Since the events $start$, $scan$, $put$, $cancel$ and $pay$ are performed by customers, they are classified as \emph{uncontrollable} for the plant $G$. This means that no supervisor can or may prevent these actions from occurring. The  event $next$ is controllable. This means that we allow a supervisor at some point not to enable the $next$-event, thereby leaving $pay$ as the only possibility to proceed.

In order to ensure that the scanner behaves in a desirable manner, a specification $R$ is displayed in the right of Figure \ref{figure-motivating1}. According to this specification, after the start and scan steps, only one state $z_2$ can  be reached. 
In this state, both the putting and cancellation options are available, which ensures that this  specification grants users the right to cancel their purchase. The transitions that can be performed in $z_4$ are similar to those in state $x_4$.

For the  similarity control problem  in the conventional sense \cite{2022Synthesis,2023MaximallyPermissive,2020Synthesis,2007Control}, there exists a $\Sigma_{uc}$-admissible supervisor $S$ shown in the left of Figure \ref{figure-motivatine example2} such that the supervised plant $S||G$ shown in the right of Figure \ref{figure-motivatine example2} is simulated by the specification $R$, that is, all the behaviors of $S||G$ are allowed by $R$.
\begin{figure}
	\centering
	\begin{minipage}[t]{0.2\textwidth}
		\centering
		\begin{tikzpicture}[
			->, >=stealth,
			every node/.style={minimum size=0.5cm, align=center},
			xnode/.style={draw, circle},
			every path/.style={thick},
			font=\scriptsize
			]

			\node[xnode] (y0) at (0, 0)    {$y_0$};
			\node[xnode] (y1) at (0, -1.2) {$y_1$};
			\node[xnode] (y2) at (0, -2.4)  {$y_2$};
			\node[xnode] (y4) at (0, -3.6)  {$y_4$};
			
			\draw (y0) -- (y1) node[midway, right] {start};
			\draw (y1) -- (y2) node[midway, left] {scan};
			\draw (y2) to[bend left=15]  node[midway, right] {cancel} (y4);
			\draw (y2) to[bend right=15] node[midway, left]  {put}    (y4);
			
			\coordinate (righty4) at ($(y4) + (1.1, 0)$);
			\draw (y4) -- (righty4) |- node[near start, right] {pay} (y0.east);
			
			\coordinate (lefty4) at ($(y4) + (-0.9, 0)$);
			\draw (y4) -- (lefty4) |- node[near start, left] {next} (y1.west);
		\end{tikzpicture}
	\end{minipage}
	\hfill
	\begin{minipage}[t]{0.25\textwidth}
		\centering
		\begin{tikzpicture}[
			->, >=stealth,
			every node/.style={align=center, inner sep=2pt},
			xnode/.style={draw, ellipse, minimum height=0.7cm, minimum width=1.4cm}, 
			every path/.style={thick},
			font=\scriptsize
			]

			\node[xnode] (y0x0) at (0, 0)    {$(y_0,x_0)$};
			\node[xnode] (y1x1) at (0, -1.2)   {$(y_1,x_1)$};
			\node[xnode] (y2x2) at (-0.8, -2.4) {$(y_2,x_2)$}; 
			\node[xnode] (y2x3) at ( 0.8, -2.4) {$(y_2,x_3)$};
			\node[xnode] (y4x4) at (0, -3.6)    {$(y_4,x_4)$};
			
			\draw (y0x0) -- (y1x1) node[midway, right] {start};
			\draw (y1x1) -- (y2x2) node[midway, left] {scan};
			\draw (y1x1) -- (y2x3) node[midway, right] {scan};
			\draw (y2x2) to[bend left=15]  node[midway, right] {cancel} (y4x4);
			\draw (y2x2) to[bend right=15] node[midway, left]  {put}    (y4x4);
			\draw (y2x3) to[bend left=15]  node[midway, right] {put}    (y4x4);

			\coordinate (righty4x4) at ($(y4x4) + (1.7, 0)$); 
			\draw (y4x4) -- (righty4x4) |- node[near start, right] {pay} (y0x0.east);
			
			\coordinate (lefty4x4) at ($(y4x4) + (-1.7, 0)$); 
			\draw (y4x4) -- (lefty4x4) |- node[near start, left] {next} (y1x1.west);
		\end{tikzpicture}
	\end{minipage}
	\caption{The supervisor $S$ (left) and supervised system $S||G$ (right)}
	\label{figure-motivatine example2}
\end{figure}
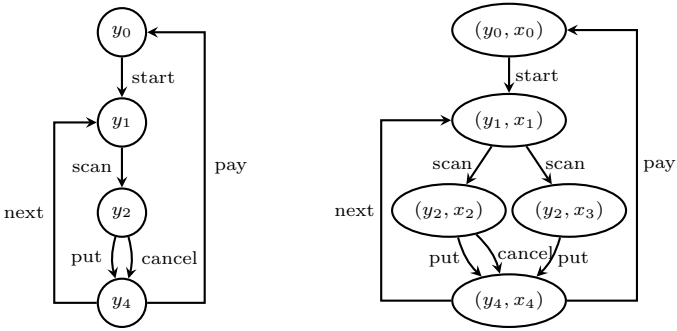
\noindent After the starting and scanning steps, the supervised system transitions nondeterministically into either the state $(y_2,x_2)$ or $(y_2,x_3)$. Since the event $cancel$ can not be initiated in the state $(y_2,x_3)$, the customer's right of cancellation transaction is not always respected by the supervised system $S||G$, and hence it doesn't support  that   customers can cancel transactions freely after scanning, which is one of the requirements from $R$.

The existing works on the similarity control problem focuses  on safety (that is, the behaviors of the supervised system must be allowed by the specification), which admits supervised systems that do not respect any requirements from the specification saying that certain transitions \emph{must} be possible. It is obvious that such a drawback disappears if, as done in \cite{SUN2014287,2006Control}, the supervised system is required to be bisimilar to the specification under consideration. However, since there exist nondeterministic transitions and uncontrollable events, a bisimilarity enforcing supervisor doesn't always exist, which inspired work on the similarity control problem \cite{2022Synthesis,2023MaximallyPermissive,2020Synthesis,2007Control}.  Nevertheless, the above example indicates that, for a given plant and specification, some solutions to the similarity control problem may be unreasonable in the sense  that the corresponding supervised systems do not realize the  requirements intended by the specification. This paper intends to filter out these unreasonable solutions by introducing \emph{required events} in similarity control problems. Roughly speaking, we will consider a supervisory control problem where the criteria for solutions fall between similarity and bisimilarity control problems.
	
For the similarity control problem with required events considered in this paper, the supervised system must not only be safe, that is, it ensures that all its behaviors are permitted by the specification, but also respect the occurrence of events that represent requirements prescribed by the specification (e.g.\ the event \textit{cancel} in the above scenario).

Similar ideas emerged in the range control problem \cite{2007Control}, where a pair of
specifications are introduced to specify the permitted and required transitions
respectively. Moreover, the standard notion of simulation is adopted to capture the relationship
between supervised systems and the pair of specifications (in detail, it should simulate one
specification and be simulated by the other). Unlike \cite{2007Control}, this paper characterizes permissions and requirements  by classifying events, and assigning different attributes to different types of events when considering the behavioral relationship between supervised systems and specifications, which is similar to the well-known modal transition systems (MTSs, for short) \cite{10.1007/3-540-52148-8_19,5119}. MTSs  are models of reactive computation based on states and transitions, where the transitions  are distinguished into two types: the $may$ transitions and the $must$ transitions.  The former describe the actions that implementations are allowed to perform, which constitute a safety property, while the latter specify actions that implementations must perform. For instance, all transitions in the specification $R$ in Figure \ref{figure-motivating1} are $may$ transitions, while those labeled by $cancel$ can additionally be designated as $must$ transitions if $R$ requires the system to respect the customer's right to cancel purchases.

To make up for the shortcomings in the research on simulation control problems mentioned above, the notion of a covariant-contravariant simulation \cite{aceto2012relating,aceto2013specification,FFP09,fabregas2010logics} will be adopted to act as our behavioral preorder. This notion divides events into three classes and proposes  different criteria  for comparing behaviors of specifications and their implementations. In such a framework, the supervised systems are required to comply with both safety requirements and required events specified by their specifications.

The main contributions  of this paper are listed as follows.

(1) This paper first formalizes the similarity control problem with required events. Then, a necessary and sufficient condition for the existence of solutions to this  problem is proposed. 

(2) Ideally, the synthesized supervisor should be as permissive as possible. Hence, we provide a method for synthesizing a maximally permissive supervisor for the similarity control problem with required events. 

(3) Since covariant-contravariant simulation is a generalization of both simulation and
  bisimulation, the results obtained in this paper are also applicable to (bi)simulation  control problems.

The rest of the paper is organized as follows. Some   preliminaries are presented in Section~\ref{sec:preliminaries}.  Section \ref{sec:controlproblem} formalizes the similarity control problem with required events and discusses necessary and sufficient  conditions for the existence of a solution to this  problem. In Section \ref{sec:synthesis}, a method for constructing a maximally permissive supervisor is proposed.
We discuss the relationship between  the similarity control problem with required events and other control problems for discrete event system in Section \ref{sec:relations}.  Section \ref{sec:conclusion} concludes this paper.
	A verification algorithm for the necessary and sufficient condition  and a synthesis algorithm for the maximally permissive supervisor are available in the appendix.
	
	\vspace{2ex}

\section{Preliminaries}\label{sec:preliminaries}

	An  automaton  is  a 4-tuple $G=(X,\Sigma, \longrightarrow_G, X_0)$, where $X$ is a set of \emph{states}, $\Sigma$ is a finite set of \emph{events}, ${\longrightarrow_G}\subseteq X\times \Sigma\times X$ and $X_0(\neq\emptyset)\subseteq X$ is the \emph{initial state set}. $G$	is said to be \emph{deterministic} if $|X_0|=1$ and for any $\sigma\in \Sigma$ and $x\in X$, there exists at most one state, say $x_1$, such that $(x,\sigma,x_1)\in {\longrightarrow_G}$.
		As usual, we write $x\stackrel{\sigma}{\longrightarrow}_Gx'$ whenever $(x,\sigma, x')\in {\longrightarrow_G}$;\linebreak[4] \plat{$x\stackrel{\sigma}{\,\not\!\!\longrightarrow}_G$ iff $(x,\sigma, x')\not\in {\longrightarrow_G}$ for any $x'\mathbin\in X$};	   moreover, $x\stackrel{\sigma}{\longrightarrow}_G$ iff $\exists x'\in X\ (x\stackrel{\sigma}{\longrightarrow}_Gx')$.

	In the standard way, ${\longrightarrow_G}\subseteq X\times \Sigma\times X$ may be generalized to ${\longrightarrow_G}\subseteq X\times \Sigma^*\times X$, where $\Sigma^*$  is the set of all finite sequences of events in $\Sigma$.

    \begin{myDef}
        Given an automaton $G=(X,\Sigma, \longrightarrow_G, X_0)$, a state $x$ is said to be  $s$-reachable $(s\in\Sigma^* )$  in an  automaton $G$ if  $x_0\stackrel{s}{\longrightarrow_G} x$  for some $x_0\in X_0$.   
	Additionally, a state $x$ is said to be  reachable in $G$ if $x$ is  $s$-reachable   for some $s\in\Sigma^*$.
    \end{myDef}
    
	 In the following,  the subscript of the transitions (e.g.\ $G$ in $\longrightarrow_G$) will be omitted if this simplification does not cause any confusion.

	\textbf{Convention:}	In the remainder of this paper, a plant, supervisor and specification are modeled as automata $G=(X,\Sigma, \longrightarrow, X_0)$, $S=(Y,\Sigma,\longrightarrow, Y_0)$ and $R=(Z,\Sigma, \longrightarrow, Z_0)$ respectively. It is assumed by default that $G$, $R$ and $S$ have the aforementioned forms when they appear.

	\begin{myDef}\label{def-composition} \cite{2006Control}
		A plant $G$ supervised by a supervisor $S$ is modeled by the synchronous composition of $S$ and $G$, defined as  $S||G=(Y\mathop\times X,\Sigma, \longrightarrow,Y_0\mathop\times X_0)$, where		
		$(y,x)\stackrel{\sigma}{\longrightarrow}(y_1,x_1) \text{ iff } x\stackrel{\sigma}{\longrightarrow}x_1 \text{ and }  y\stackrel{\sigma}{\longrightarrow}y_1$ for any 	$(y,x), (y_1,x_1)\in Y\times X \text{ and } \sigma\in \Sigma.$
		
	\end{myDef}

	In supervisory control theory \cite{kushi2017synthesis,2023MaximallyPermissive,1987Supervisory,SUN2014287,2020Synthesis,2007Control,2006Control},  events are divided into two classes: \emph{controllable}  and \emph{uncontrollable} events. We denote by $\Sigma_{uc}$ (or, $\Sigma_{c}$) the set of uncontrollable (resp.\ controllable) events, so $\Sigma=\Sigma_{uc}\cup\Sigma_{c}$ and $\Sigma_{uc}\cap\Sigma_{c}=\emptyset$.	
	In addition, since this paper aims to deal with similarity control problems involving required events, we denote the set of required events as $\Sigma_r$ $ (\subseteq\Sigma)$.
	
	In order to formalize the intuition that uncontrollable events may not be disabled by  supervisors, the notion of $\Sigma_{uc}$-admissibility is adopted in   \cite{Kimura2014Maximally,kushi2017synthesis,2022Synthesis,2023MaximallyPermissive}, as recalled below.     
	\begin{myDef}\label{def-admissible uc}
		Given a plant $G$  and a specification $R$, a supervisor $S$ is said to be  \emph{$\Sigma_{uc}$-admissible} with respect to $G$ if, for any reachable state $(y,x)$ of $S||G$  and $\sigma\in \Sigma_{uc}$, $$(y,x)\stackrel{\sigma}{\longrightarrow} \text{whenever } x\stackrel{\sigma}{\longrightarrow}.$$	
	\end{myDef}
	
	\begin{myDef}\label{def-simulation}
		Let $G$ and $R$  be two  automata and $\Phi\subseteq X\times Z$. 
		
		\noindent $(1)$ $\Phi$ is said to be a \emph{simulation relation} from $G$ to $R$ if it satisfies the following two conditions:
		
		 (initial state) for any $x_0\in X_0$, there exists $z_0\in Z_0$ such that $(x_0,z_0)\in \Phi$ and
		
		 (forward) for any $(x,z)\in \Phi$, $x'\in X$ and $\sigma\in \Sigma$, $x\stackrel{\sigma}{\longrightarrow}x'$ implies $z\stackrel{\sigma}{\longrightarrow}z'$ and $(x',z')\in \Phi$ for some $z'\in Z$. 
		
		\noindent $(2)$ $\Phi$  is said to be a \emph{covariant-contravariant simulation  relation} (\emph{cc-simulation} for short) if it is a simulation relation such that

		($\Sigma_r$-backward) for any $ (x,z)\in \Phi$, $z' \in Z'$ and $\sigma\in \Sigma_r$,
		$z\stackrel{\sigma}{\longrightarrow}z'$ implies 
		$x\stackrel{\sigma}{\longrightarrow}x'$ and $ (x',z')\in \Phi$ for some $x'\in X$.	
		
		\noindent $(3)$ $\Phi$ is said to be a \emph{bisimulation} if both $\Phi$ and its inverse relation\footnote{Given a binary relation $\Phi\subseteq X\times Z$, its inverse relation $\Phi^{-1}\subseteq Z\times X$ is defined as $\Phi^{-1}=\{(z,x):(x,z)\in \Phi\}$.}$\Phi^{-1}$ are simulations.

		\noindent $(4)$  $\Phi$ is said to  be a \emph{$\Sigma_{ucr}$-simulation relation} from $G$ to $R$ if  it satisfies the conditions  (initial state),  ($\Sigma_r$-backward) and
		
		($\Sigma_{uc}$-forward) for any $ (x,z)\in \Phi$, $x'\in X$ and $\sigma\in \Sigma_{uc}$, $x\stackrel{\sigma}{\longrightarrow}x'$ implies $z\stackrel{\sigma}{\longrightarrow}z'$ and $ (x',z')\in \Phi$ for some $z'\in Z$.

		\noindent $(5)$ $\Phi$ is said to be a \emph{$\Sigma_{uc}$-simulation relation}  if it satisfies the conditions  (initial state) and  ($\Sigma_{uc}$-forward). 	
	\end{myDef}
	
	If there exists a  simulation   (or, cc-simulation, bisimulation, $\Sigma_{ucr}$-simulation, $\Sigma_{uc}$-simulation) relation $\Phi$ from $G$ to $R$, then $G$ is said to be simulated  (cc-simulated, bisimulated, $\Sigma_{ucr}$-simulated, $\Sigma_{uc}$-simulated) by $R$, in symbols, $G\sqsubseteq   R$  ($G\sqsubseteq_{cc}   R$, $G\simeq   R$, $G\sqsubseteq_{ucr}   R$, $G\sqsubseteq_{uc}   R$ resp.). It is easy to see that  $\sqsubseteq $, $\sqsubseteq_{cc}  $,  $\sqsubseteq_{ucr}  $ and $\sqsubseteq_{uc}  $ are preorders, that is, they are transitive and reflexive. Moreover, $\simeq$ is an equivalence relation  (reflexive, symmetric, and transitive).
	We also adopt the notation $\Phi: G\sqsubseteq   R$ to indicate that $\Phi$ is a  simulation   relation from $G$ to $R$. Similar notations are used for $\sqsubseteq_{cc}$, $\simeq$, $\sqsubseteq_{ucr}$ and $\sqsubseteq_{uc}$.
	
	The notion of cc-simulation in Definition \ref{def-simulation} is a special case of the classical  covariant-contravariant simulation  relation defined in \cite{aceto2012relating, aceto2013specification, FFP09, fabregas2010logics}, namely with the bivariant action set $\Sigma_r$,  covariant action set $\Sigma-\Sigma_r$ and  contravariant action set $\emptyset$.
	Note that taking $\Sigma_r=\emptyset$ makes cc-simulation equal to simulation, whereas taking $\Sigma_r=\Sigma$ makes cc-simulation almost equal to bisimulation (equal if $|Z_0|=1$).	
		Normally, simulation (or, bisimulation) is a finer behavioral notion than language inclusion (trace equivalence, resp.), that is, the former implies the latter. It is almost common knowledge that, for any two automata $R_1$ and $R_2$, $$R_1\simeq R_2 \Longrightarrow  R_1\sqsubseteq_{cc}R_2 \Longrightarrow R_1\sqsubseteq R_2 \Longrightarrow L(R_1)\subseteq L(R_2)$$ $$\text{ and }R_1\simeq R_2\Longrightarrow L(R_1)=L(R_2),$$ where $L(R_i)=\{s\mathbin\in\Sigma^*\!\!: z_0\stackrel{s}{\longrightarrow}\!\! \text{ for some initial state $z_0$ in }R_i\}$ with $i\in\{1,2\}$. Here the first two implications immediately follow from Definition 3, and the others can be found in \cite{GLABBEEK20013}.\vspace{2ex}

	\section{The Similarity Control Problem with Required Events}\label{sec:controlproblem}
	
    In this section, a new control problem for DESs  is first formulated. Subsequently, a necessary and sufficient condition for the existence of a solution to this problem is established.

    As  mentioned in Section 1, the current formalization of the similarity control problem for discrete event systems may be inadequate in certain scenarios as it fails to deal with the demand mandated by the specification. To compensate for this deficiency, the notion of  cc-simulation relation will be adopted in this paper to formalize  the control problem. In detail, the similarity control problem with required events is described formally as follows. 

	\textbf{Problem Formulation}	Given a plant $G$ and specification $R$, the $(G,R)$-similarity control problem with required events  is to find a $\Sigma_{uc}$-admis\-sible  (w.r.t.\ $G$)   supervisor $S$ such that $S||G\sqsubseteq_{cc} R$.

    The set of all such supervisors is denoted by $\SPR (G,R)$.
	This problem is said to be solvable if such  supervisors exist. Compared with  usual similarity control problems, the supervisor $S$ above is expected to meet more requirements, so that the supervised system $S||G$ respects the demands of the specification $R$ appropriately, which is captured formally by the condition  ($\Sigma_r$-backward) involved in the notion of cc-simulation.

	\begin{example}
		Consider the motivating example given in Section 1 again. In order to  respect the customers' right to cancel the purchase, we can set  $cancel$ as a required event  (i.e.\ $\Sigma_r=\{\text{cancel}\}$). 	In this scenario, since the state $(y_2,x_3)$ is reachable in $S||G$, each (cc-)simulation from $S||G$ to $R$ must involve its (cc-)simulation counterpart. Moreover, it's not difficult to see that $z_2$ is the unique state that can act as the simulation counterpart of the state $(y_2,x_3)$. Further, since each cc-simulation  relation must satisfy ($\Sigma_r$-backward), it follows from $\Sigma_r=\{cancel\}$, $z_2\stackrel{cancel}{\longrightarrow}$ and $(y_2,x_3)\,\,\,\,\not\!\!\!\!\stackrel{cancel}{\longrightarrow}$ that $S||G\not\sqsubseteq_{cc}R$, and hence
		the supervisor $S$ presented in Figure \ref{figure-motivatine example2} is not a solution for the $ (G,R)$-similarity control problem with required events.
	\end{example}

	For simplicity, we introduce the following predicate. Its schematic   is presented in Figure \ref{figure-introduction of match}, where $(x,z)\in W$ and $(x_1,z_1),(x_2,z_2)\in W'$.
	
	\begin{myDef}\label{def-the predication}
		Given two automata $G$ and $R$, a ternary relation $match_{G,R}$ over $\powerset (X\times Z) \times \Sigma \times \powerset (X\times Z)$\footnote{The powerset \powerset(X) of a set $X$ is defined as $\powerset(X)=\{A:A\subseteq X\}$.} is defined by, for any $W,W'\in \powerset (X\times Z)$ and $\sigma\in \Sigma$, $match_{G,R} (W,\sigma,W')$ iff $W$ satisfies the condition ($\sigma$-forward) up to $W'$, in detail,		
		$$match_{G,R} (W,\sigma,W') \Leftrightarrow \forall (x,z)\in W \, \forall x'\in X $$ $$(x\stackrel{\sigma}{\longrightarrow}x'\Longrightarrow \exists z'\in Z (z\stackrel{\sigma}{\longrightarrow}z' \text{ and }  (x',z')\in W')).$$ 
	\end{myDef}
	
	 \begin{figure}
\centering
\begin{minipage}[t]{0.48\linewidth}
    \centering
    \begin{tikzpicture}[
        ->, >=stealth,
        every node/.style={minimum size=0.5cm, align=center},
        xnode/.style={draw, circle},
        every path/.style={thick},
        font=\small,
        baseline=(x),    
        remember picture
        ]
        
        \node[xnode] (x) at (0, 0)    {$x$};
        \node[xnode] (x1) at (-1, -1) {$x_1$};
        \node[xnode] (x2) at (1, -1)  {$x_2$};       
 \node (xdots) at (1.3, -0.5) {$\dots$};

        \draw (x) -- (x1) node[midway, left] (t_label) {$\sigma$};
        \draw (x) -- (x2) node[midway, right] (v_label) {$\sigma$}; 
    \end{tikzpicture}
\end{minipage}
\hfill
\begin{minipage}[t]{0.48\linewidth}
    \centering
    \begin{tikzpicture}[
        ->, >=stealth,
        every node/.style={minimum size=0.5cm, align=center},
        xnode/.style={draw, circle},
        every path/.style={thick},
        font=\small,
        baseline=(z),    
        remember picture
        ]
        
        \node[xnode] (z) at (0, 0)    {$z$};
        \node[xnode] (z1) at (-1, -1)   {$z_1$};
        \node[xnode] (z2) at (1, -1) {$z_2$};
  \node (xdots) at (1.3, -0.5) {$\dots$};

        \draw[dashed] (z) -- (z1) node[midway, left] (u_label) {$\sigma$};
        \draw[dashed] (z) -- (z2) node[midway, right] (w_label) {$\sigma$};   
      
    \end{tikzpicture}
\end{minipage}

\begin{tikzpicture}[remember picture, overlay]
    
    \draw[dashed, thick] (x) to[bend left=10] node[above] {$W$} (z);

    \draw[dashed, thick] (x1) to[bend right=15] node[below] {$W'$} (z1);
       \draw[] (x2) to[bend right=15] node[below] {$W'$} (z2);

\end{tikzpicture}
\vspace{2ex}

\caption{The figure of $match_{G,R}(W,\sigma,W')$}
\label{figure-introduction of match}
\end{figure}

		In order to provide a necessary and sufficient condition for the existence of a solution to the similarity control problem with required events, a notion called $\Sigma_{ucr}$-controllability set is defined as follows.

	\begin{myDef}\label{def-controlability}
		Let $G$ and $R$  be two  automata. A  set  $E\subseteq \powerset (X\times Z)$ is said to be a \emph{$\Sigma_{ucr}$-controllability set} from $G$ to $R$  if it satisfies the following conditions:
		
		\noindent	 (istate)	there exists $W_0\in E$ such that $\forall x_0\in X_0 \, \exists z_0\in Z_0  ( (x_0,z_0)\in W_0)$;

		\noindent$ (\ref{def-controlability}-a) $ for any $W\in E$ and $\sigma\in \Sigma_{uc}$, there exists $W'\in E$ with  $match_{G,R} (W,\sigma,W')$;

		\noindent	$ (\ref{def-controlability}-b)$  for any $W\in E$, $ (x,z)\in W$, $\sigma\in \Sigma_r$ and $z'\in Z$ with $z\stackrel{\sigma}{\longrightarrow}z'$, there exists $x'\in X$ and $W'\in E$ such that $x\stackrel{\sigma}{\longrightarrow}x',   (x',z')\in W' $ and $match_{G,R} (W,\sigma,W')$.

	\end{myDef}
	
	Given automata $G$ and $R$, the next lemma states that the class of all $\Sigma_{ucr}$-controllability sets from $G$ to $R$ is closed under the operator $P:\powerset(\powerset(X\times Z))\longrightarrow \powerset(\powerset(X\times Z))$ defined by $P(E)=\bigcup_{W\in E}\powerset(W)$ for any $E\in \powerset(\powerset(X\times Z))$, that is, for any $W'\subseteq X\times Z$, $W'\in P(E)$ iff $W'\subseteq W$ for some $W\in E$.

	\begin{lemma}\label{lemma-E^* is successor}
			Let $G$ and $R$  be two  automata and  $E$  a $\Sigma_{ucr}$-controllability set from $G$ to $R$. Then $P(E)=\bigcup_{\widetilde{W}\in E}\powerset (\widetilde{W})$ satisfies the clause (istate) in Definition \ref{def-controlability} and the following two conditions. Hence, $P(E)$  is a $\Sigma_{ucr}$-controllability set from $G$ to $R$.

        	\noindent$ (\ref{def-controlability}-a') $ for any $W\in P(E)$ and $\sigma\in \Sigma_{uc}$, there exists $W'\in P(E)$ such that  $match_{G,R} (W,\sigma,W')$ and $$W'\subseteq \bigcup_{ (x,z)\in W} \{x': x \stackrel{\sigma}{\longrightarrow}x'\}\times \{z': z \stackrel{\sigma}{\longrightarrow}z'\};$$
		
		\noindent	$ (\ref{def-controlability}-b')$  for any $W\in P(E)$, $ (x,z)\in W$, $\sigma\in \Sigma_r$ and $z'\in Z$ with $z\stackrel{\sigma}{\longrightarrow}z'$, there exists $x'\in X$ and $W'\in P(E)$ such that $x\stackrel{\sigma}{\longrightarrow}x',   (x',z')\in W' $,  $match_{G,R} (W,\sigma,W')$ and $$W'\subseteq \bigcup_{ (x,z)\in W} \{x': x \stackrel{\sigma}{\longrightarrow}x'\}\times \{z': z \stackrel{\sigma}{\longrightarrow}z'\}.$$
			
		\begin{proof}
			Since $E$ is a 	$\Sigma_{ucr}$-controllability set from $G$ to $R$ and $E\subseteq P(E)$, the condition  (istate) holds trivially for $P(E)$. To complete the proof, it
			suffices to verify that $P(E)$ satisfies the conditions 	$ (\ref{def-controlability}-a') $  and 	$ (\ref{def-controlability}-b') $. In the following, we deal with the former, while the latter can be verified similarly---this is omitted here.
			Assume that $W^*\in P(E)$ and $\sigma\in \Sigma_{uc}$. Then there exists $W_1\in E$ such that $W^*\subseteq W_1$.
			Since $E$ is a $\Sigma_{ucr}$-controllability set from $G$ to $R$ and $W_1\in E$, there exists $W''\in E$ such that 
			$$\forall (x,z)\in W_1 \, \forall x'\in X  (x\stackrel{\sigma}{\longrightarrow}x'\Longrightarrow \exists z'\in Z $$ $$(z\stackrel{\sigma}{\longrightarrow}z' \text{ and }  (x',z')\in W'')).$$
			Set
			$$W'= W''\cap\bigcup_{ (x,z)\in W^*}\{x':x\stackrel{\sigma}{\longrightarrow} x'\}\times\{z':z\stackrel{\sigma}{\longrightarrow} z'\}.$$ 
			Clearly, it also holds that
			$$\forall (x,z)\in W^* \, \forall x'\in X  (x\stackrel{\sigma}{\longrightarrow}x'\Longrightarrow \exists z'\in Z $$
			$$(z\stackrel{\sigma}{\longrightarrow}z' \text{ and }  (x',z')\in W')),$$
			moreover, $W'\in P(E)$ due to $W'\subseteq W''\in E$, as desired.	
		\end{proof}				
	\end{lemma}

	Before giving a sufficient and necessary condition for the existence of a solution to the similarity control problem with required events,  a method for constructing automata from $\Sigma_{ucr}$-controllability sets is provided.
	\begin{myDef}\label{def-supervisor SE}
		Let $G$, $R$  be two  automata and $E$ a $\Sigma_{ucr}$-controllability set from $G$ to $R$. The automaton  $\mathcal{A} (E)\mathbin= ( \bigcup_{\widetilde{W}\in E}\powerset (\widetilde{W}),\Sigma, \longrightarrow, I_E)$ is defined by
		$$ I_E=\{W_0\in \bigcup_{\widetilde{W}\in E}\powerset (\widetilde{W}): \forall x_0\in X_0 \, \exists z_0\in Z_0 \  ( (x_0,z_0)\in W_0) $$ $$\text{ and }W_0\subseteq X_0\times Z_0\} $$		
		and, for any $W,W'\in \bigcup_{\widetilde{W}\in E}\powerset (\widetilde{W})$ and $\sigma\in \Sigma$, $W\stackrel{\sigma}{\longrightarrow}W'$ iff the following holds:

		\noindent$ (\ref{def-supervisor SE}-a)$ $\exists (x,z)\in W$ $\exists (x',z')\in W'$ $ (x\stackrel{\sigma}{\longrightarrow}x' \text{ and } z\stackrel{\sigma}{\longrightarrow}z')$;

		\noindent$ (\ref{def-supervisor SE}-b)$ $match_{G,R} (W,\sigma,W')$;
		
		\noindent$ (\ref{def-supervisor SE}-c)$ $W'\in\powerset (\bigcup_{ (x,z)\in W}\{x':x\stackrel{\sigma}{\longrightarrow} x'\}\times\{z':z\stackrel{\sigma}{\longrightarrow} z'\})$.

	\end{myDef}
	
	Since $\bigcup_{W\in E}\powerset (W)$ is a $\Sigma_{ucr}$-controllability set from $G$ to $R$ due to Lemma \ref{lemma-E^* is successor}, by Definition \ref{def-controlability}, there exists a set $W_0\in \bigcup_{W\in E}\powerset (W)$  realizing the clause  (istate) in Definition \ref{def-controlability}, which ensures that 
	the initial state set $I_E$ of $\mathcal{A} (E)$  is nonempty.

	\begin{example}\label{example-construction of SE}
		Consider the plant $G=(X,\Sigma,\longrightarrow,\{x_0\})$ and specification $R=(Z,\Sigma,\longrightarrow,\{z_0\})$ in Figure \ref{exampleoftheorem1}, where  $\Sigma_{uc}=\{uc_1,uc_2\}$, $\Sigma_{c}=\{c\}$ and $\Sigma_r=\{l\}$.

\begin{figure}
	\centering
  \begin{minipage}{0.32\columnwidth} \centering 

		\centering
		\begin{tikzpicture}[
			->, >=stealth,
			every node/.style={minimum size=0.4cm, align=center},
			xnode/.style={draw, circle},
			every path/.style={thick},
			font=\scriptsize
			]
			
			\node[xnode] (x0) at (0, 0)    {$x_0$};
			\node[xnode] (x1) at (0, -1.2)   {$x_1$};
			\node[xnode] (x2) at (-0.6, -2.4) {$x_2$};
			\node[xnode] (x3) at (0.6, -2.4) {$x_3$};
			\node[xnode] (x4) at (0, -3.6)    {$x_4$};

			\draw (x0) -- (x1) node[midway, right] {c};
			\draw (x1) -- (x2) node[midway, left] {uc$_1$};
			\draw (x1) -- (x3) node[midway, right] {uc$_1$};
			\draw (x2) to[bend left=15]  node[midway, right] {$l$} (x4);
			\draw (x2) to[bend right=15] node[midway, left]  {uc$_1$} (x4);
			\draw (x3) to[bend left=15]  node[midway, right] {$l$} (x4);

			\coordinate (rightx4) at ($(x4) + (1.0, 0)$);
			\draw (x4) -- (rightx4) |- node[near start, right] {uc$_2$} (x0.east);
		\end{tikzpicture}
	\end{minipage}
	\hfill
  \begin{minipage}{0.35\columnwidth} \centering 

		\centering
		\begin{tikzpicture}[
			->, >=stealth,
			every node/.style={minimum size=0.4cm, align=center},
			xnode/.style={draw, circle},
			every path/.style={thick},
			font=\scriptsize
			]

			\node[xnode] (z0) at (0, 0)    {$z_0$};
			\node[xnode] (z1) at (0, -1.2)   {$z_1$};
			\node[xnode] (z2) at (-0.9, -2.4) {$z_2$};
			\node[xnode] (z3) at (0.9, -2.4) {$z_3$};
			\node[xnode] (z4) at (0, -3.6)    {$z_4$};
			
			\draw (z0) -- (z1) node[midway, right] {c};
			\draw (z1) -- (z2) node[midway, left] {uc$_1$};
			\draw (z1) -- (z3) node[midway, right] {uc$_1$};
			\draw (z2) to[bend left=15]  node[midway, right] {$l$} (z4);
			\draw (z2) to[bend right=15] node[midway, left]  {uc$_1$} (z4);
			\draw (z3) to[bend right=15]  node[midway, left] {$l$} (z4);
			\draw (z3) to[bend left=15]  node[midway, right] {uc$_1$} (z4);
			
			\coordinate (rightz4) at ($(z4) + (1.3, 0)$);
			\draw (z4) -- (rightz4) |- node[near start, right] {uc$_2$} (z0.east);
		\end{tikzpicture}
	\end{minipage}
	\hfill
  \begin{minipage}{0.28\columnwidth} \centering 

		\centering
		\begin{tikzpicture}[
			->, >=stealth,
			every node/.style={minimum size=0.4cm, align=center},
			xnode/.style={draw, circle},
			every path/.style={thick},
			font=\scriptsize
			]

			\node[xnode] (w0) at (0, 0)    {$W_0$};
			\node[xnode] (w1) at (0, -1.05) {$W_1$};
			\node[xnode] (w2) at (0, -2.1)  {$W_2$};
			\node[xnode] (w4) at (0, -3.3)  {$W_3$};
			
			\draw (w0) -- (w1) node[midway, right] {c};
			\draw (w1) -- (w2) node[midway, right] {uc$_1$};
			\draw (w2) to[bend left=15]  node[midway, right] {$l$} (w4);
			\draw (w2) to[bend right=15] node[midway, left]  {uc$_1$} (w4);

			\coordinate (rightw4) at ($(w4) + (0.8, 0)$);
			\draw (w4) -- (rightw4) |- node[near start, right] {uc$_2$} (w0.east);
		\end{tikzpicture}
	\end{minipage}
	\caption{The plant $G$ (left), specification $R$ (middle) and reachable part of supervisor $\mathcal{A}(E)$ (right)}
			\label{exampleoftheorem1}
\end{figure}
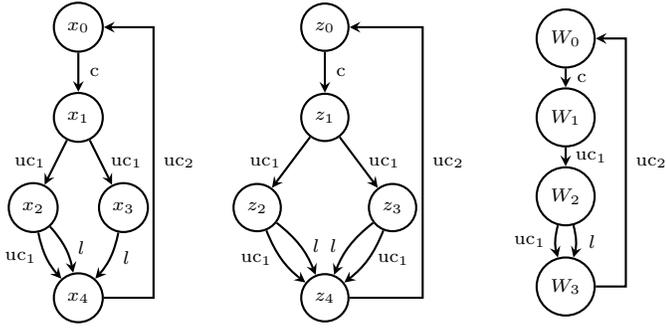
		
		It could be checked straightforwardly that the set 
		$E=\{W_i: 0\leqslant i \leqslant 3 \}$ is a $\Sigma_{ucr}$-controllability set $E$ from $G$ to $R$ and $\bigcup_{W\in E}\powerset (W)$ $=E\cup\{W_4,W_5,\emptyset\}$,
		where $W_0=\{ (x_0,z_0)\}$, $W_1=\{ (x_1,z_1)\}$, $W_2=\{ (x_2,z_2), (x_3,z_3)\}$, 
		$W_3=\{ (x_4,z_4)\}$, $W_4=\{ (x_2,z_2)\}$	and $W_5=\{ (x_3,z_3
		)\}$. By Definition \ref{def-supervisor SE}, the reachable part of the supervisor $\mathcal{A} (E)$ is   presented graphically in Figure \ref{exampleoftheorem1}.

	\end{example}
	
	The following lemma reveals that  $\mathcal{A} (E)$ contains sufficient information about the transitions of the automaton $G$. Specifically, the $G$-projection of any $s$-reachable state in the former consists exactly of all $s$-reachable states in the latter.

	\begin{lemma}\label{lemma-about xz in W}
		Given two  automata  $G$ and $R$, if $E$ is a $\Sigma_{ucr}$-controllability set from $G$ to $R$ and $\mathcal{A} (E)= (\bigcup_{\widetilde{W}\in E}\powerset (\widetilde{W}),\Sigma,\linebreak \longrightarrow, I_E)$, then, for any $s\mathbin\in \Sigma^*$, $W_0\mathbin\in I_E$ and $W\mathbin\in\bigcup_{\widetilde{W}\in E}\powerset (\widetilde{W})$, $W_0\stackrel{s}{\longrightarrow}W$ implies $\pi_G (W)=Reach (s,X_0)$, where 
		\begin{align}
			&\pi_G (W)= \{x\in X:  (x,z)\in W \text{ for some }z\in Z\} \text{ and }\notag\\ &Reach (s,X_0)=\{x\in X: x_0\stackrel{s}{\longrightarrow}x \text{ for some }x_0\in X_0\}\notag.
		\end{align}

		\begin{proof}
			It proceeds by induction on $|s|$, i.e., the length of $s$. If $|s|=0$ then $s=\varepsilon$ and $W_0=W$, and hence $\pi_G (W)=X_0=Reach (s,X_0)$ by  Definition \ref{def-supervisor SE}. In the following, we deal with the inductive step and let $|s|=k+1$ with $k\geqslant 0$.   		 		
			Assume that $s=s_1\sigma$ and $W_0\stackrel{s_1}{\longrightarrow}W_1\stackrel{\sigma}{\longrightarrow}W$ with  $W_1\in \bigcup_{\widetilde{W}\in E}\powerset (\widetilde{W})$. By the induction hypothesis, we have $\pi_G (W_1)=Reach (s_1,X_0)$. Next we intend to show $\pi_G (W)=Reach (s,X_0)$. 
			
			($\subseteq$) Let $x\in \it left.$ Then there exists $z\in Z$ such that $ (x,z)\in W$. Due to $W_1\stackrel{\sigma}{\longrightarrow}W$, by the clause $ (\ref{def-supervisor SE}-c)$ in Definition \ref{def-supervisor SE}, we get\vspace{-5pt}
   $$W\in\powerset (\bigcup_{ (x_1,z_1)\in W_1}\{x':x_1\stackrel{\sigma}{\longrightarrow} x'\}\times\{z':z_1\stackrel{\sigma}{\longrightarrow} z'\}).$$  	
			Thus, there exists $ (x_1,z_1)\in W_1$ such that $x_1\stackrel{\sigma}{\longrightarrow}x$ and $z_1\stackrel{\sigma}{\longrightarrow}z$. Due to $ (x_1,z_1)\mathbin\in W_1$, $x_1\mathbin\in \pi_G (W_1)\mathbin=Reach (s_1,X_0)$. Then it follows that $x_0\stackrel{s_1}{\longrightarrow}x_1\stackrel{\sigma}{\longrightarrow}x$ for some $x_0\in X_0$, which implies $x\in Reach (s,X_0)$.
			
			$ (\supseteq)$ Let $x\in right.$ Thus, $x_0\stackrel{s_1}{\longrightarrow}x_1\stackrel{\sigma}{\longrightarrow}x$ for some $x_1\in X$ and $x_0\in X_0$. Then we have $x_1\in Reach (s_1,X_0)=\pi_G (W_1)$, and thus $ (x_1,z_1)\in W_1$ for some $z_1\in Z$. Due to $W_1\stackrel{\sigma}{\longrightarrow}W$, by the clause $ (\ref{def-supervisor SE}-b)$ in Definition \ref{def-supervisor SE}, it follows from 
			$x_1\stackrel{\sigma}{\longrightarrow}x$ and $ (x_1,z_1)\in W_1$ that
			$z_1\stackrel{\sigma}{\longrightarrow}z$ and $ (x,z)\in W$ for some $z\in Z$, and hence $x\in \pi_G (W)$.	  
		\end{proof}
	\end{lemma}
	
	The next lemma reveals that
	the  supervisor $\mathcal{A} (E)$ is a solution for the $ (G,R)$-similarity control problem with required events.

	\begin{lemma}\label{lemma- E exists implies AE is a solution}
		Let $G$ and $R$ be two automata. For any $\Sigma_{ucr}$-controllability set $E$ from $G$ to $R$, $\mathcal{A} (E)\in \SPR (G,R)$.\\
		
		\begin{proof}
			Let  $\mathcal{A} (E)= (\bigcup_{\widetilde{W}\in E}\powerset (\widetilde{W}),\Sigma,\longrightarrow,I_E)$.
			First, we prove that $\mathcal{A} (E)$ is $\Sigma_{uc}$-admissible. Let $ (W,x)$ be  reachable in $\mathcal{A} (E)||G$, $\sigma\in \Sigma_{uc}$ and $x\stackrel{\sigma}{\longrightarrow}$. Then $ (W_0,x_0)\stackrel{s}{\longrightarrow}  (W,x)$ for some $s\in \Sigma^*$, $W_0\in I_E$  and $x_0\in X_0$, and hence $W_0\stackrel{s}{\longrightarrow}W$ and $x_0\stackrel{s}{\longrightarrow}x$ by Definition \ref{def-composition}. Further, by Lemma \ref{lemma-about xz in W}, we get $x\in Reach (s,X_0)= \pi_G (W)$, and thus
			$ (x,z)\in W$ for some $z\in Z$. Since $\bigcup_{\widetilde{W}\in E}\powerset (\widetilde{W})$ is a $\Sigma_{ucr}$-controllability set  from $G$ to $R$ due to Lemma \ref{lemma-E^* is successor}, by the clause $ (\ref{def-controlability}-a')$ in  Lemma \ref{lemma-E^* is successor}, there exists $W'\in \bigcup_{\widetilde{W}\in E}\powerset (\widetilde{W})$ such that the clauses $ (\ref{def-supervisor SE}-b)$ and $ (\ref{def-supervisor SE}-c)$ in Definition \ref{def-supervisor SE} hold for $\sigma$, $W$ and $W'$, which, together with facts that $ (x,z)\in W$, $x\stackrel{\sigma}{\longrightarrow}$ and $\sigma\in \Sigma$, implies that $ (\ref{def-supervisor SE}-a)$ also holds for $\sigma$, $W$  and $W'$. Therefore, $W\stackrel{\sigma}{\longrightarrow}W'$ due to Definition \ref{def-supervisor SE}. Thus, by Definition \ref{def-admissible uc}, $\mathcal{A} (E)$ is $\Sigma_{uc}$-admissible.

			Next we show that $\mathcal{A} (E)||G\sqsubseteq_{cc}R$. To this end, we prove that
			the relation $\Phi_{E}\subseteq  (\bigcup_{\widetilde{W}\in E}\powerset (\widetilde{W})\linebreak[1]\times X)\times Z$ defined as\vspace{-3pt}
			$$ \Phi_{E}=\{ ( (W,x),z):  (x,z)\in W \text{ and } W\in \bigcup_{\widetilde{W}\in E}\powerset (\widetilde{W})\}\vspace{-3pt}$$ 
			is a cc-simulation relation from $\mathcal{A} (E)||G$ to $R$. Let $ (W_0,x_0)$ be any initial state in $\mathcal{A} (E)||G$. Then $ (W_0,x_0)\in I_E\times X_0$. By Definition \ref{def-supervisor SE}, there exists $z_0\in Z_0$ such that $ (x_0,z_0)\in W_0\in \bigcup_{\widetilde{W}\in E}\powerset (\widetilde{W})$. Hence, $ ( (W_0,x_0),z_0)\in \Phi_E$. Thus, the  condition  (initial state) in Definition \ref{def-simulation} holds.  Below, we verify the conditions  ($\Sigma$-forward) and  ($\Sigma_r$-backward) in Definition \ref{def-simulation} in turn.	Let $ ( (W,x),z)\in \Phi_E$. Then $ (x,z)\in W\in\bigcup_{\widetilde{W}\in E}\powerset (\widetilde{W})$.
			
			($\Sigma$-forward) Assume that $ (W,x)\stackrel{\sigma}{\longrightarrow} (W',x')$ and   $\sigma\in \Sigma$. By the clause $ (\ref{def-supervisor SE}-b)$ in Definition \ref{def-supervisor SE}, it follows from $ (x,z)\in W$ and $x\stackrel{\sigma}{\longrightarrow}x'$  that $z\stackrel{\sigma}{\longrightarrow}z'$ and $ (x',z')\in W'$ for some $z'\in Z$. Clearly, $ ( (W',x'),z')\in \Phi_E$.

			($\Sigma_r$-backward) Assume $z\stackrel{\sigma}{\longrightarrow}z'$  and $\sigma\in \Sigma_r$.
			Then, by  the clause $ (\ref{def-controlability}-b')$ in  Lemma \ref{lemma-E^* is successor}, it follows from $ (x,z)\in W\in \bigcup_{\widetilde{W}\in E}\powerset (\widetilde{W})$, $\sigma\in \Sigma_r$ and $z\stackrel{\sigma}{\longrightarrow}z'$ that there exists $x'\in X$ and $W'\in E$ such that $x\stackrel{\sigma}{\longrightarrow}x'$, $ (x',z')\in W'\subseteq \bigcup_{ (x,z)\in W}\{x':x\stackrel{\sigma}{\longrightarrow} x'\}\times\{z':z\stackrel{\sigma}{\longrightarrow} z'\}$ and  	$match_{G,R} (W,\sigma,W')$. 
			Thus, the clauses $ (\ref{def-supervisor SE}-a)$, $ (\ref{def-supervisor SE}-b)$ and  $ (\ref{def-supervisor SE}-c)$ in Definition \ref{def-supervisor SE} hold for $\sigma$, $W$ and $W'$, and hence $W\stackrel{\sigma}{\longrightarrow}W'$, which implies $ (W,x)\stackrel{\sigma}{\longrightarrow} (W',x')$ due to $x\stackrel{\sigma}{\longrightarrow}x'$. Moreover, $ ( (W',x'),z')\in \Phi_E$ because of $ (x',z')\in W'\in \bigcup_{\widetilde{W}\in E}\powerset (\widetilde{W})$, as desired.   		
		\end{proof}
	\end{lemma}

	The above lemma asserts that any $\Sigma_{ucr}$-controllability set brings a solution to the $ (G,R)$-similarity control problem with required events
	 through the operator $\mathcal{A} (\cdot)$. In the following, we intend to show that any solution also induces a $\Sigma_{ucr}$-controllability set. To this end, the operator $E (\cdot)$ is introduced, which transforms a cc-simulation relation into a controllability set. Given $\Phi\!:S||G\sqsubseteq_{cc} R$, $E(\Phi)$ is a collection indexed by $Y$, and each element $\theta_y$ of $E(\Phi)$ consists of pairs determined by $y$ and $\Phi$.

	\begin{myDef}\label{def-EPhi}	Let $S$, $G$ and $R$  be  automata and $\Phi\!:S||G\sqsubseteq_{cc} R$. The set $E (\Phi)$ is defined as
		$E (\Phi)=\{\theta_y:y\in Y\}$, where\vspace{-4pt}		
		$$\theta_y=\{ (x,z): ( (y, x),z)\in \Phi \text{ and } (y,x) \text{ is reachable in } S||G\}\vspace{-4pt} $$\noindent\text{ for each }$y\in Y$. 
	\end{myDef} 
	
	Roughly speaking, the above definition stems from the fact that, under suitable conditions, any set like $\theta_y$ has a $\sigma$-match like $\theta_{y'}$ (i.e.\ $match_{G,R}(\theta_y,\sigma,\theta_{y'})$ holds), which is the most non-trivial property involved in the notion of a $\Sigma_{ucr}$-controllability set.
	Fortunately, the following lemma asserts that the set $E (\Phi)$ is indeed a $\Sigma_{ucr}$-controllability set  from $G$ to $R$ for each automaton $S$ and  $\Phi:S||G \sqsubseteq_{cc} R$.
	
	\begin{lemma}\label{lemma-S implies exists ES}
		Given the plant $G$ and specification $R$, if $S$ is a $\Sigma_{uc}$-admissible supervisor and $\Phi:S||G\sqsubseteq_{cc} R$, then $E(\Phi)$  is a $\Sigma_{ucr}$-controllability set from $G$ to $R$.\\
		
		\begin{proof}
			Let $S=(Y,\Sigma,\longrightarrow,Y_0)$. Since $Y_0\neq \emptyset$, we can choose arbitrarily and fix a state $y_0\in Y_0$. Clearly, $\theta_{y_0}\in E (\Phi)$. Moreover, for  any
			$x_0\in X_0$,
			due to $ (y_0,x_0)\in Y_0\times X_0$ and $\Phi: S||G\sqsubseteq_{cc} R$, there exists $z_0\in Z_0$ such that $ ( (y_0,x_0),z_0)\in \Phi$, and hence $ (x_0,z_0)\in \theta_{y_0}$. Thus, the set $\theta_{y_0}$ satisfies the condition  $\forall x_0\in X_0 \, \exists z_0\in Z_0 \, ( (x_0,z_0)\in \theta_{y_0})$, and hence $\theta_{y_0}$ realizes  the  clause  (istate) in Definition \ref{def-controlability}. 
			It remains to show that $E (\Phi)$ satisfies the condition $ (\ref{def-controlability}-a) $ and $ (\ref{def-controlability}-b) $ in Definition \ref{def-controlability}.

			$ (\ref{def-controlability}-a)$ Let $W\in E (\Phi)$ and $\sigma\in \Sigma_{uc}$. Hence,	$W= \theta_y$ for some $y\in Y$.
			If $x\stackrel{\sigma}{\not\longrightarrow}$ for each $ (x,z)\in W$, then we set $W'= W$, and it is easy to see that $match_{G,R}(W,\sigma,W')$  holds trivially. 
			Next we consider the other case that $x\stackrel{\sigma}{\longrightarrow}$ for some $ (x,z)\in W$. We choose arbitrarily and fix such a pair $ (x,z)\in W$. Then $ ( (y, x),z)\in \Phi$ and $(y, x)$ is reachable in $S||G$.	Further, since $S$ is  $\Sigma_{uc}$-admissible, by Definition~\ref{def-admissible uc}, $y\stackrel{\sigma}{\longrightarrow}$ due to $\sigma\in \Sigma_{uc}$ and $x\stackrel{\sigma}{\longrightarrow}$. Thus, we can choose arbitrarily and fix  a state $y^*\in Y$ with $y\stackrel{\sigma}{\longrightarrow}y^*$ and put\vspace{-2pt}
			$$W'= \theta_{y^*}.\vspace{-2pt}$$ 
			Clearly, $W'\in E (\Phi)$.	
			In the following, we intend to verify $match_{G,R} (W,\sigma,W')$. Let  $ (u,v)\in W$ with $u\stackrel{\sigma}{\longrightarrow}u^*$. Thus, $ ( (y, u),v)\in \Phi$ and $(y, u)$ is reachable in $S||G$. Then $ (y,u)\stackrel{\sigma}{\longrightarrow} (y^*,u^*)$ due to $y\stackrel{\sigma}{\longrightarrow}y^*$ and $u\stackrel{\sigma}{\longrightarrow}u^*$. Further, it follows from  $ ( (y, u),v)\in \Phi$ and $\Phi: S||G\sqsubseteq_{cc} R$ that $ ( (y^*, u^*),v^*)\in \Phi$ for some $v^*$ with $v\stackrel{\sigma}{\longrightarrow}v^*$, and thus we have $ (u^*,v^*)\in \theta_{y^*}$ $ (=W')$, as desired.\\\mbox{}

			$ (\ref{def-controlability}-b) $	Let $W\in E (\Phi)$, $ (x,z)\in W$,  $z\stackrel{\sigma}{\longrightarrow}z_1$ and $\sigma\in\Sigma_r$. Then $W= \theta_y$ for some $y\in Y$. Thus,
			$ (x,z)\in \theta_y$, and hence $ ( (y, x),z)\in \Phi$. Since $ ( (y,x),z)\in \Phi$, $\sigma\in \Sigma_r$ and $z\stackrel{\sigma}{\longrightarrow}z_1$, we have $ (y,x)\stackrel{\sigma}{\longrightarrow} (y_1,x_1)$ for some $ (y_1,x_1)\in Y\times X$ with $ ( (y_1,x_1),z_1)\in \Phi $. Thus, $ (x_1,z_1)\in \theta_{y_1}$.
			Put\vspace{-2pt}
			$$W'= \theta_{y_1}.\vspace{-2pt}$$  		 
			Clearly, $x\stackrel{\sigma}{\longrightarrow}x_1$ and $ (x_1,z_1)\in W'\in E (\Phi)$.
			To complete this proof, it is enough to show  $match_{G,R} (W,\sigma,W')$.
			Let $ (u,v)\in W  (=\theta_y)$ and $u\stackrel{\sigma}{\longrightarrow}u_1$. Hence, $ ( (y,u),v)\in \Phi$, and $ (y,u)\stackrel{\sigma}{\longrightarrow} (y_1,u_1)$ because of $y\stackrel{\sigma}{\longrightarrow}y_1$ and $u\stackrel{\sigma}{\longrightarrow}u_1$. Further, due to $ ( (y,u),v)\in \Phi$, we get $v\stackrel{\sigma}{\longrightarrow}v_1$ for some  $v_1\in Z$ with $ ( (y_1,u_1),v_1)\in \Phi$. Clearly, $ (u_1,v_1)\in W'$, as desired.
		\end{proof}
	\end{lemma}
	
	Now we arrive at the main result of this section, which provides   a  necessary and sufficient condition for the solvability  of  the similarity control problem with required events.
	
	\begin{theorem}\label{th-uc admissbile is sufficient}
		Let $G$ and $R$
		be two automata. There exists a $\Sigma_{uc}$-admissible supervisor $S$ such that $S||G\sqsubseteq_{cc}   R$ iff there exists a $\Sigma_{ucr}$-controllability set $E$ from $G$ to $R$.\\
		
		\begin{proof}
			This immediately follows from Lemma \ref{lemma- E exists implies AE is a solution} and \ref{lemma-S implies exists ES}.
		\end{proof}
	\end{theorem}

	Given  $W,W'\in E\subseteq \powerset (X\times Z)$ and
	$\sigma\in \Sigma$, the time complexity of verifying $match_{G,R}(W,\sigma,W')$ is
	$O(|X|^2|Z|^2)$. Hence, for each $W\mathbin\in E$, the time complexity of testing  whether $W$ satisfies the clauses $ (\ref{def-controlability}-a)$ (or, $ (\ref{def-controlability}-b))$ is
	$$O(|X|^2|Z|^2\times|\Sigma_{uc}|\times|E|)$$ $$(\text{resp.\ }
	O(|X|^2|Z|^2\times|\Sigma_r|\times|E|\times |X|^2|Z|^2)).$$
	
	\noindent Thus the time complexity of verifying whether  $W$ satisfies the clauses 
	$ (\ref{def-controlability}-a)$ and $ (\ref{def-controlability}-b)$ is
	$O(|X|^4|Z|^4\times|(\Sigma_{uc}+\Sigma_r)|\times|E|).$
	Since there are at most $2^{|X||Z|}$ elements in $E$, the time complexity of deciding the existence of a solution for a $ (G,R)$-similarity control problem with required events is at most
	$$	O (|X|^4|Z|^4\times(|\Sigma|+|\Sigma_r|) \times 2^{2|X||Z|}). $$

	\begin{example}
		Consider the motivating example shown in Figure \ref{figure-motivating1} with $\Sigma_{uc}=\{start, scan, put, cancel, pay\}$, $\Sigma_{c}=\linebreak[4]\{next\}$ and $\Sigma_r=\{cancel\}$. Assume that 
		there  exists a $\Sigma_{ucr}$-controllability set $E$ from $G$ to $R$. Then we have $(x_0,z_0)\in W_0\in E$ for some $W_0$, and hence there exists $W\in E$ such that $(x_1,z_1)\in W$ and $match_{G,R}(W_0,start,W)$ due to the clause $(\ref{def-controlability}-a)$ in Definition \ref{def-controlability}. Further, since $scan\in \Sigma_{uc}$, due to $(\ref{def-controlability}-a)$ in  Definition \ref{def-controlability} again, there exists $W_1\in E$ such that $match_{G,R} (W,scan,W_1)$. For any such $W_1\in E$, it follows from Definitions \ref{def-the predication} and \ref{def-controlability} that $ (x_2,z_2),  (x_3,z_2)\in W_1$.	 
		However, for $(x_3,z_2)\in W_1$ and $z_2\stackrel{cancel}{\longrightarrow}z_4$ with $cancel\in \Sigma_r$, the clause $ (\ref{def-controlability}-b)$ does not hold for $W_1$ due to $x_3\,\,\not\!\!\stackrel{cancel}{\longrightarrow}$.  Thus, there is no $\Sigma_{ucr}$-controllability set from $G$ to $R$, and hence this $ (G,R)$-similarity control problem with required events is unsolvable.
	\end{example}
	
	\vspace{-2ex}
	
	\begin{example}		
		Consider the plant $G$, specification $R$ and supervisor $\mathcal{A} (E)$ in Example \ref{example-construction of SE}.  The reachable part of the supervised plant $\mathcal{A} (E)||G$ is presented graphically in Figure \ref{exampleoftheorem2}.
		Clearly, $\mathcal{A} (E)$ is a $\Sigma_{uc}$-admissible supervisor and $\mathcal{A} (E)||G\sqsubseteq_{cc} R$.

		\begin{figure}
			\centering
			\centering
			\begin{tikzpicture}[				
					->, >=stealth,
			every node/.style={align=center, inner sep=2pt},
			xnode/.style={draw, ellipse, minimum height=0.7cm, minimum width=1.4cm}, 
			every path/.style={thick},
			font=\scriptsize
			]
				
				\node[xnode] (W0x0) at (0, 0)    {$(W_0, x_0)$};
				\node[xnode] (W1x1) at (0, -1.2)   {$(W_1, x_1)$};
				\node[xnode] (W2x2) at (-1.5, -2.4) {$(W_2, x_2)$};
				\node[xnode] (W2x3) at ( 1.5, -2.4) {$(W_2, x_3)$};
				\node[xnode] (W3x4) at (0, -3.6)    {$(W_3, x_4)$};
				
				\draw (W0x0) -- (W1x1) node[midway, right] {c};
				\draw (W1x1) -- (W2x2) node[midway, left] {uc$_1$};
				\draw (W1x1) -- (W2x3) node[midway, right] {uc$_1$};
				\draw (W2x2) to[bend left]  node[midway, right] {$l$} (W3x4);
				\draw (W2x2) to[bend right] node[midway, left]  {uc$_1$}    (W3x4);
				\draw (W2x3) to[bend left]  node[midway, right] {$l$}    (W3x4);
				
				\coordinate (belowW2x3) at ($(W3x4) + (2.4, 0)$);
				\coordinate (rightW0x0) at ($(W0x0) + (2.4, 0)$);
				
				\draw (W3x4)
				-- (belowW2x3)
				-- node[midway, right]{uc$_2$}
				(rightW0x0)
				-- (W0x0.east);
				
			\end{tikzpicture}

			\caption{The reachable part of supervised system $\mathcal{A} (E)||G$}
			\label{exampleoftheorem2}
		\end{figure}
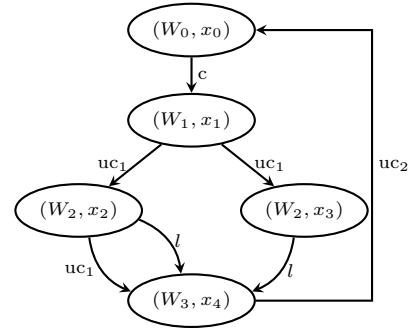
	\end{example}

	In \cite{2020Synthesis,wang2024more}, it has been shown that there exists a $\Sigma_{uc}$-admissible supervisor $S$ such that $S||G\sqsubseteq R$ if and only if $G\sqsubseteq_{uc}R$, which characterizes the solvability of similarity control problem in terms of the coinductive notion of $\Sigma_{uc}$-simulation. For similarity control problems with required events as considered in this paper, inspired by this result, a reasonable conjecture arises, namely that the notion of $\Sigma_{ucr}$-simulation, which is derived from the notion of $\Sigma_{uc}$-simulation and addresses the $\Sigma_r$-requirement by adding the clause ($\Sigma_r$-backward) from Definition \ref{def-simulation},  can also be used to characterize the solvability of these problems. In 
	the remaining part of this section we will explain that, except for the case of deterministic automata, this notion alone is not sufficient to characterize the solvability of similarity-control problems with required events. We begin with making the following  observation, which follows from Definitions \ref{def-simulation} and \ref{def-controlability}.

	\begin{observation}\label{obs-control set implies ucr simulation}
		Given two  automata  $G$ and $R$, if $E$ is a $\Sigma_{ucr}$-controllability set from $G$ to $R$, then $\bigcup E$ is a $\Sigma_{ucr}$-simulation relation from $G$ to $R$.\\
		
		\begin{proof}
		 Since $E$ is a $\Sigma_{ucr}$-controllability set from $G$ to $R$, by the clause (istate) in Definition \protect\ref{def-controlability}, we have
			$$\forall x_0\in X_0 \, \exists z_0\in Z_0  ( (x_0,z_0)\in W_0)  $$ 
			
			\noindent for some $W_0\in E\subseteq \bigcup E$, which implies that the condition (initial state) in Definition \protect{\ref{def-simulation}} holds. It remains to verify ($\Sigma_{uc}$-forward) and ($\Sigma_r$-backward) in turn. Let $(x,z)\in \bigcup E$. Then $(x,z)\in W$ for some $W\in E$.
			
			($\Sigma_{uc}$-forward) Assume that $x\stackrel{\sigma}{\longrightarrow}x'$ and $\sigma\in \Sigma_{uc}$.  Since $E$ is a $\Sigma_{ucr}$-controllability set and $(x,z)\in W\in E$, it follows from the clause $ (\ref{def-controlability}-a) $ in Definition  \protect\ref{def-controlability} that $match_{G,R}(W,\sigma,\!W')$ for some $W'\in E$, that is, there exists $z'\in Z$ such that $z\stackrel{\sigma}{\longrightarrow}z'$ and $(x',z')\in W'\subseteq \bigcup E$.
			
			($\Sigma_r$-backward) Assume that $z\stackrel{\sigma}{\longrightarrow}z'$ and $\sigma\in \Sigma_r$. By $(x,z)\in W\in E$ and the clause $ (\ref{def-controlability}-b) $ in Definition \protect\ref{def-controlability}, there exists $x'\in X$ and $W'\in E$ such that $x\stackrel{\sigma}{\longrightarrow}x'$ and $(x',z')\in W'\subseteq \bigcup E$, as desired.	
		\end{proof}
	\end{observation}
	
	Thus, a $\Sigma_{ucr}$-controllability set $E$ is a cover\footnote{Given a set $X$, a set $E\subseteq \powerset(X)$ is a cover of $X$ if $X=\bigcup E$.} of the 
	$\Sigma_{ucr}$-simulation relation $\bigcup E$ which satisfies the conditions  (istate), $ (\ref{def-controlability}-a)$ and $ (\ref{def-controlability}-b)$ in Definition \ref{def-controlability}. Hence,  Theorem \ref{th-uc admissbile is sufficient} can be expressed equivalently as: there exists a $\Sigma_{uc}$-admissible supervisor such that $S||G\sqsubseteq_{cc} R$ iff there is a $\Sigma_{ucr}$-simulation relation from $G$ to $R$ which has a cover satisfying  (istate), $ (\ref{def-controlability}-a)$ and $ (\ref{def-controlability}-b)$.	
	The next example illustrates that not every $\Sigma_{ucr}$-simulation has such a cover.
	
	\begin{example}\label{example-cover doest not always exist}
		Consider the plant $G=(X,\Sigma,\longrightarrow,\{x_0\})$ and specification $R=(Z,\Sigma,\longrightarrow,\{z_0\})$ shown in Figure~\ref{exampleofnondetermin} with $\Sigma=\Sigma_r=\{l,l_1\}$, $\Sigma_{c}=\{l\}$ and $\Sigma_{uc}=\{l_1\}$.
	
		\begin{figure}[!ht]
			\centering
			\begin{minipage}[t]{0.2\textwidth}
				\centering
				\begin{tikzpicture}[
					->, >=stealth,
					every node/.style={minimum size=0.5cm, align=center},
					xnode/.style={draw, circle},
					every path/.style={thick},
					font=\scriptsize
					]
					
					\node[xnode] (x0) at (0, 0)     {$x_{0}$};
					\node[xnode] (x1) at (-1.2, -1.2) {$x_{1}$};
					\node[xnode] (x2) at (1.2, -1.2)  {$x_{2}$};
					\node[xnode] (x3) at (-1.2, -2.4) {$x_{3}$};
					
					\draw (x0) -- (x1) node[midway, left]  {$l$};
					\draw (x0) -- (x2) node[midway, right] {$l$};
					\draw (x1) -- (x3) node[midway, right] {$l_1$};
				\end{tikzpicture}
			\end{minipage}
			\hfill
			\begin{minipage}[t]{0.2\textwidth}
				\centering
				\begin{tikzpicture}[
					->, >=stealth,
					every node/.style={minimum size=0.5cm, align=center},
					xnode/.style={draw, circle},
					every path/.style={thick},
					font=\scriptsize
					]
					
					\node[xnode] (z0) at (0, 0)    {$z_{0}$};
					\node[xnode] (z1) at (0, -1.2) {$z_{1}$};
					\node[xnode] (z2) at (0, -2.4) {$z_{2}$};
					
					\draw (z0) -- (z1) node[midway, right] {$l$};
					\draw (z1) -- (z2) node[midway, right] {$l_1$};
				\end{tikzpicture}
			\end{minipage}
			\caption{The plant $G$ (left) and specification $R$ (right)} 
			\label{exampleofnondetermin} 
		\end{figure}
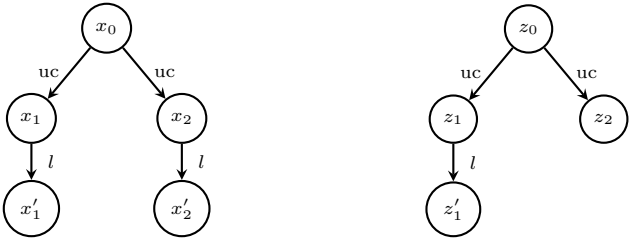
		
		Clearly, there exists a $\Sigma_{ucr}$-simulation relation  from $G$ to $R$, e.g., 	$\Phi=\{(x_0,z_0),(x_1,z_1),(x_3,z_2)\}.$
		Assume that there exists a $\Sigma_{ucr}$-controllability set  $E$ from $G$ to $R$. Then we have $\{(x_0,z_0)\}\in W_0\in E$ for some $W_0$.  Hence, by the clause $(\ref{def-controlability}-b)$ in Definition \ref{def-controlability} and \plat{$z_0\stackrel{l}{\longrightarrow}z_1$} with $l\in \Sigma_r$, there  exists $x\in X$ and $W\in E$ such that \plat{$x_0\stackrel{l}{\longrightarrow}x$}, $(x,z_1)\in W$ and $match_{G,R}(W_0,l,W)$. Now we obtain $(x_1,z_1), (x_2,z_1)\in W\!$.
		However, due to \plat{$x_2\not\stackrel{~l_1}{\longrightarrow}$}, there is no $x'$ and $W'$ realizing the clause $(\ref{def-controlability}-b)$ for $l_1\in \Sigma_r$, $z_1\stackrel{l_1}{\longrightarrow}z_2$ and $(x_2,z_1)\in W\in E$. Thus, there is no $\Sigma_{ucr}$-controllability set $E$, that is, not every $\Sigma_{ucr}$-simulation from $G$ to $R$ has a cover satisfying (istate), $(\ref{def-controlability}-a)$ and $(\ref{def-controlability}-b)$.
	\end{example}
	
	However, for $\Sigma_{ucr}$-simulations between deterministic automata, covers satisfying  (istate), $ (\ref{def-controlability}-a)$ and $ (\ref{def-controlability}-b)$ always exist. The notion below strengthens the notion of $\Sigma_{ucr}$-simulation by imposing an additional requirement on the $\sigma$-forward matching for $\sigma\in\Sigma_r$.
	
	\begin{myDef}
		$\!$A  $\Sigma_{ucr}$-simulation $\hspace{-1pt}\Phi$ is said to be \emph{uniform} w.r.t.\ $\Sigma_r$ if, for any $ (x_i,z_i)\in \Phi$ $ (i=1,2)$ such that  $x_1$, $x_2$, $z_1$ and $z_2$ are $s$-reachable in $G$ and $R$ respectively for some $s\in \Sigma^*$, 
		$$\forall\sigma\in \Sigma_r\,\forall x_2' (z_1\stackrel{\sigma}{\longrightarrow} \text{ and } x_2\stackrel{\sigma}{\longrightarrow}x_2' \Longrightarrow$$ $$\exists z_2' (z_2\stackrel{\sigma}{\longrightarrow}z_2' \text{ and }  (x_2',z_2')\in \Phi )).$$
	\end{myDef}

    The next lemma shows that a $\Sigma_{ucr}$-simulation, which is uniform w.r.t. $\Sigma_r$, can induce a  $\Sigma_{ucr}$-controllability set.
	
	\begin{lemma}\label{lemma-ucr simulation implies control set}
		Let $G$ and $R$ be two automata. If $\Phi:G\sqsubseteq_{ucr} R$ is uniform w.r.t.\ $\Sigma_r$, then there exists a $\Sigma_{ucr}$-controllability set from $G$ to $R$.
		
		\begin{proof}
			For any $s\in \Sigma^*$, put $E=\{W_s:s\in \Sigma^*\}$ and
			$$W_s= \{ (x,z)\in \Phi: x,z \text{ are }s \text{-reachable in } G \text{ and }R \text{ resp.}\}.$$
			We intend to show that $E$ is  a $\Sigma_{ucr}$-controllability set from $G$ to $R$.	Clearly, $\forall x_0\in X_0 \, \exists z_0\in Z_0\linebreak[2] ( (x_0,z_0)\in W_\epsilon\in E)$. Thus, $E$ satisfies  (initial state) in Definition \ref{def-controlability}. 
			We check the conditions $ (\ref{def-controlability}-a,b) $ in turn.
			
			$ (\ref{def-controlability}-a) $ For any $s\in \Sigma^*$ and $\sigma\in \Sigma_{uc}$, it is easy to see that $match_{G,R} (W_s,\sigma, W_{s\sigma})$ because $\Phi$ satisfies the condition  ($\Sigma_{uc}$-forward)  in Definition \ref{def-simulation}.

			$ (\ref{def-controlability}-b) $ Let $ (x,z)\in W_s$, $\sigma\in \Sigma_r$ and $z'\in Z$ with $z\stackrel{\sigma}{\longrightarrow}z'$. Then $ (x,z)\in \Phi$ and $x,z$ are $s$-reachable in $G$ and $R$.
			By the condition  ($\Sigma_r$-backward) in Definition \ref{def-simulation}, $x\stackrel{\sigma}{\longrightarrow}x'$ and $ (x',z')\in \Phi$ for some $x'\in X$, and hence $ (x',z')\in W_{s\sigma}$. Assume that $ (u,v)\in W_s$ and $u\stackrel{\sigma}{\longrightarrow}u'$. Then $ (u,v)\in \Phi$.
			Since $\Phi$ is uniform and $z\stackrel{\sigma}{\longrightarrow}z'$, we get $v\stackrel{\sigma}{\longrightarrow}v'$ and $ (u',v')\in \Phi$ for some $v'\in Z$. Thus $ (u',v')\in W_{s\sigma}$ because $u'$ and $v'$ are $s\sigma$-reachable in $G$ and $R$ respectively, which implies that $match_{G,R} (W_s,\sigma,W_{s\sigma})$ holds.
		\end{proof}
	\end{lemma}

	The following example shows that $\bigcup E$ needs be not uniform w.r.t.\ $\Sigma_r$ even if $E$ is a $\Sigma_{ucr}$-controllability set.
	
	\begin{example}
		Consider the plant $G=(X,\Sigma,\longrightarrow,\{x_0\})$ and $R=(Z,\Sigma,\longrightarrow,\{z_0\})$ shown in Figure~\ref{exampleofuniform} with $\Sigma_{uc}=\{\mbox{uc}\}$ and  $\Sigma_{c}=\Sigma_r=\{l\}$. It is easy to check that 
		\begin{align*}\nonumber			E=\{&\{(x_0,z_0)\},\{(x_1,z_1),(x_2,z_1)\},\\&\{(x_1',z_1'),(x_2',z_1')\},\{(x_2,z_2)\}\}
		\end{align*}

		\noindent is a $\Sigma_{ucr}$-controllability set from $G$ to $R$.

\begin{figure}[!ht]
	\centering
	\begin{minipage}[t]{0.2\textwidth}
		\centering
		\begin{tikzpicture}[
			->, >=stealth,
			every node/.style={minimum size=0.5cm, align=center},
			xnode/.style={draw, circle},
			every path/.style={thick},
			font=\scriptsize
			]
			
			\node[xnode] (x0) at (0, 0)      {$x_{0}$};
			\node[xnode] (x1) at (-1.0, -1.2) {$x_{1}$};
			\node[xnode] (x2) at (1.0, -1.2)  {$x_{2}$};
			\node[xnode] (x1') at (-1.0, -2.4) {$x_{1}'$};
			\node[xnode] (x2') at (1.0, -2.4)  {$x_{2}'$};
			
			\draw (x0) -- (x1) node[midway, left]  {uc};
			\draw (x0) -- (x2) node[midway, right] {uc};
			\draw (x1) -- (x1') node[midway, right] {$l$};
			\draw (x2) -- (x2') node[midway, right] {$l$};
		\end{tikzpicture}
	\end{minipage}
	\hfill
	\begin{minipage}[t]{0.2\textwidth}
		\centering
		\begin{tikzpicture}[
			->, >=stealth,
			every node/.style={minimum size=0.5cm, align=center},
			xnode/.style={draw, circle},
			every path/.style={thick},
			font=\scriptsize
			]
			
			\node[xnode] (z0) at (0, 0)      {$z_{0}$};
			\node[xnode] (z1) at (-1.0, -1.2) {$z_{1}$};
			\node[xnode] (z2) at (1.0, -1.2)  {$z_{2}$};
			\node[xnode] (z1') at (-1.0, -2.4) {$z_{1}'$};
			
			\draw (z0) -- (z1) node[midway, left]  {uc};
			\draw (z0) -- (z2) node[midway, right] {uc};
			\draw (z1) -- (z1') node[midway, right] {$l$};
		\end{tikzpicture}
	\end{minipage}
	\caption{The plant $G$ (left) and specification $R$ (right)} 
	\label{exampleofuniform} 
\end{figure}

		By Observation \ref{obs-control set implies ucr simulation}, $\bigcup E$ is a  $\Sigma_{ucr}$-simulation relation from $G$ to $R$. However, for the $uc$-reachable states $x_1,x_2,z_1,z_2$, since $(x_1,z_1)$, $(x_2,z_2)\mathbin\in \bigcup E$, \plat{$z_1\!\stackrel{l}{\longrightarrow}\! \text{ and } x_2\stackrel{l}{\longrightarrow}x_2'$} but \plat{$z_2\not\!\stackrel{~l}{\longrightarrow}$}, $\bigcup E$ is not uniform w.r.t.\ $\Sigma_r$.
	\end{example}
	\vspace{-3ex}

	\begin{lemma}\label{lemma-determin implies uniform}
		For any deterministic automata $G$ and $R$, if $\Phi:G\sqsubseteq_{ucr}R$ then $\Phi$ is uniform w.r.t.\ $\Sigma_r$.\\
		
		\begin{proof}
			This immediately follows from the fact that, for any $s\in \Sigma^*$, there exists at most one $s$-reachable state in each deterministic automaton.
		\end{proof}
	\end{lemma}
	
	The following corollary follows immediately from Observation \ref{obs-control set implies ucr simulation}, Lemma \ref{lemma-ucr simulation implies control set} and \ref{lemma-determin implies uniform}  and  Theorem \ref{th-uc admissbile is sufficient}. It provides another necessary and sufficient condition for the  solvability of the similarity control problem with required events under the circumstance that the relevant  plant and specification  are deterministic. Example \ref{example-cover doest not always exist} indicates that this restriction is necessary.

	\begin{corollary}\label{cor-control set is the same as ucr simulation}
		For any deterministic automata $G$  and $R$, the $ (G,R)$-similarity control problem with required events is solvable iff $G\sqsubseteq_{ucr} R$.
	\end{corollary}

	\section{Synthesis of Maximally Permissive Supervisors}\label{sec:synthesis}

This section intends to provide a method for synthesizing a maximally permissive  supervisor  (w.r.t.\  the  simulation relation) for the similarity control problem with required events.
     Let $\MSPR (G,R)$ be the set of all maximally
	permissive supervisors  for  the similarity control
	problem with required events, formally,
	\begin{align*}
		\MSPR (G, R)= \{&S\in \SPR (G, R): \forall S'\in \SPR (G, R)\\
		&(S'||G\sqsubseteq S||G)\}.
	\end{align*}	
	\noindent Note that, as usual, we adopt the classical notion of simulation to characterize the relative fineness of solutions.

     To synthesize the maximally permissive supervisor, a fixpoint characteristic of $\Sigma_{ucr}$-controllability sets is provided first.
    Then  the candidate maximally permissive supervisor is constructed by using the operator $\mathcal{A}(\cdot)$ given in Definition \ref{def-supervisor SE} and the fixpoint.

	\begin{myDef}\label{def-function}
		Let $G$ and  $R$ be automata. A fixpoint characterization function $F_{ (G,R)}:\powerset (\powerset (X\times Z))\longrightarrow \powerset (\powerset (X\times Z))$ is defined by, for each $E\subseteq \powerset (X\times Z)$ and $W\subseteq X\times Z$, $W\in F_{ (G,R)} (E)$ iff $W\in E$ and $W$ satisfies the following conditions:
		
		\noindent (1) $\forall\sigma\in \Sigma_{uc} \, \exists W'\in E \, ( match_{G,R} (W,\sigma,W') )$;
		
		\noindent (2) $\forall  (x,z)\in W \, \forall\sigma\in \Sigma_r \, \forall z'\in Z \, (z\stackrel{\sigma}{\longrightarrow}z'\Longrightarrow\exists x'\mathbin\in X$ \, $\exists W'\in E \, (
		x\stackrel{\sigma}{\longrightarrow}x',   (x',z')\in W'\text{ and }match_{G,R} (W,\sigma,W') ).$	  
	\end{myDef}
	
	 It is not difficult to see that $W\in F_{ (G,R)} (E)$ if and only if $W$ satisfies the clauses $(\ref{def-controlability}-a)$ and $(\ref{def-controlability}-b)$ in Definition \ref{def-controlability}.
	Clearly, $F_{ (G,R)}$ is a monotonic function over the complete lattice \allowbreak $ (\powerset (\powerset (X\times Z)), \subseteq)$, that is, for any $E, E'\subseteq \powerset (X\times Z)$,
	$$E\subseteq E' \Longrightarrow F_{ (G,R)} (E)\subseteq F_{ (G,R)} (E').$$
	Hence, by Tarski's fixpoint theorem \cite{tarski1955lattice}, there exists the greatest fixpoint of the function $F_{ (G,R)}$, denoted by $E_{ (G,R)}^\uparrow$.

The following conclusion reveals the connection between  $\Sigma_{ucr}$-controllability sets and fixpoints of $F_{ (G,R)}$.

	\begin{proposition}\label{prop-fixpoint}
		Let $G$ and  $R$ be two automata. 
		
		\noindent$ (1)$ For any $E\subseteq \powerset (X\times Z)$,  $E$ is a $\Sigma_{ucr}$-controllability set from $G$ to $R$ iff $F_{ (G,R)} (E)=E$ and there exists  $W_0\in E$  such that\vspace{-3pt}
		$$\forall x_0\in X_0 \, \exists z_0\in Z_0 \  ( (x_0,z_0)\in W_0).$$
		
		\noindent$ (2)$ There exists a $\Sigma_{ucr}$-controllability set from $G$ to $R$ iff there exists  $W_0\in E_{ (G,R)}^\uparrow$ such that\vspace{-3pt}
		$$ \forall x_0\in X_0 \, \exists z_0\in Z_0 \  ( (x_0,z_0)\in W_0).$$

		\begin{proof}
			The first clause immediately follows from Definitions \ref{def-controlability} and \ref{def-function}. We prove the second one. 
			
			($\Longleftarrow$) Since $E_{ (G,R)}^\uparrow$ is 
	        the		
			greatest fixpoint of $F_{ (G,R)}$, we get $F_{ (G,R)}(E_{ (G,R)}^\uparrow)=E_{ (G,R)}^\uparrow$. Then by $(1)$, $E_{ (G,R)}^\uparrow$ is  a $\Sigma_{ucr}$-controllability set from $G$ to $R$. 
			
			 	($\Longrightarrow$)		 Let $E$ be a $\Sigma_{ucr}$-controllability set from $G$ to $R$. Then it follows from $(1)$ that $F_{ (G,R)} (E)=E$, and hence $E\subseteq E_{ (G,R)}^\uparrow$.
			Due to the clause (istate) in Definitions \ref{def-controlability},  there exists  $W_0\in E$ such that
			$\forall x_0\in X_0 \, \exists z_0\in Z_0 \  ( (x_0,z_0)\in W_0).$ 
		Thus, $W_0\in E_{ (G,R)}^\uparrow$.
		\end{proof}
	\end{proposition}
	
	The next corollary provides some
     properties of the greatest fixpoint of $F_{ (G,R)}$ whenever $\SPR(G,R)\neq\emptyset$,  which will be used later.
	
	\begin{corollary}\label{corollary- fixpoint is control set}  Let $G$ be a plant and $R$ a specification.
    
   \noindent $ (1)$ $\SPR (G,R)\mathbin{\neq}\emptyset$ iff  there exists  $W_0\mathbin\in E_{ (G,R)}^\uparrow$  such that   
		$\forall x_0\in X_0 \, \exists z_0\in Z_0 \  ( (x_0,z_0)\in W_0).$
				
		\noindent $(2)$	$E_{ (G,R)}^\uparrow=\bigcup_{W\in E_{ (G,R)}^\uparrow}\powerset (W)$ whenever $\SPR (G,R)\neq \emptyset$.\\
		
		\begin{proof}
			$ (1)$ By Proposition \ref{prop-fixpoint} and Theorem \ref{th-uc admissbile is sufficient}, $\SPR (G,R)\neq\emptyset$ iff $E_{ (G,R)}^\uparrow$ is a $\Sigma_{ucr}$-controllability set from $G$ to $R$, and thus $ (1)$ holds. 	
			
			$ (2)$ Due to  $\SPR (G,R)\neq \emptyset$, by clause $ (1)$ and Proposition \ref{prop-fixpoint},  $E_{ (G,R)}^\uparrow$
			is a $\Sigma_{ucr}$-controllability set from $G$ to $R$. Hence, by Lemma \ref{lemma-E^* is successor}, $\bigcup_{W\in E_{ (G,R)}^\uparrow}\powerset (W)$ is a $\Sigma_{ucr}$-controllability set from $G$ to $R$. Thus, by Proposition \ref{prop-fixpoint},\vspace{-3pt} $$F_{ (G,R)} (\bigcup_{W\in E_{ (G,R)}^\uparrow}\powerset (W))=\bigcup_{W\in E_{ (G,R)}^\uparrow}\powerset (W).\vspace{-3pt}$$ Hence,
			we get\vspace{-3pt}
			$$\bigcup_{W\in E_{ (G,R)}^\uparrow}\powerset (W)\subseteq E_{ (G,R)}^\uparrow.\vspace{-3pt}$$ Moreover, $E_{ (G,R)}^\uparrow\subseteq\bigcup_{W\in E_{ (G,R)}^\uparrow}\powerset (W)$ holds trivially, as desired.
		\end{proof}
	\end{corollary}

	For any $\Sigma_{ucr}$-controllability set $E$, Definition \ref{def-supervisor SE} provides a method for constructing the automaton  $\mathcal{A} (E)$ which is a solution to the similarity control problem with required events. In fact, with $E_{ (G,R)}^\uparrow$, $\mathcal{A} (E_{ (G,R)}^\uparrow)$ is a maximally permissive supervisor	whenever $\SPR (G,R)\neq\emptyset$. We will verify this conclusion below. The following lemma reflects the monotonicity of $\mathcal{A} (\cdot)||G$.

	\begin{lemma}\label{lemma-AE smaller than fixpoint}
	 If both $E_1$ and $E_2$ are $\Sigma_{ucr}$-controllability sets from the automaton $G$ to $R$ with $E_1\subseteq E_2$, then  $\mathcal{A} (E_1)||G\sqsubseteq \mathcal{A} (E_2)||G$.\\
		
		\begin{proof}
			By Definition \ref{def-supervisor SE}, we may assume that\vspace{-3pt} $$\mathcal{A} (E_i)= (\bigcup_{\widetilde{W}\in E_i}\powerset (\widetilde{W}),\Sigma,\longrightarrow,I_{E_i})\  (i=1,2).\vspace{-6pt}$$ Since  $E_1\subseteq E_2$, we have $\bigcup_{\widetilde{W}\in E_1}\powerset (\widetilde{W})\subseteq \bigcup_{\widetilde{W}\in E_2}\powerset (\widetilde{W})$.\pagebreak[2]
			Then it is routine to check that the relation $$\Phi\subseteq  (\bigcup_{\widetilde{W}\in E_1}\powerset (\widetilde{W})\times X)\times  (\bigcup_{\widetilde{W}\in E_2}\powerset (\widetilde{W})\times X)\vspace{-3pt}$$ defined as 
			$$ \Phi=\{ ( (W,x), (W,x)): W\in \bigcup_{\widetilde{W}\in E_1}\powerset (\widetilde{W}) \text{ and } x\in X\}$$ 
			is a simulation relation from $\mathcal{A} (E_1)||G$ to $\mathcal{A} (E_2)||G$ by the observation that $I_{E_1}\!\subseteq \!I_{E_2}$ and, for any $W,\!W'\!\in\!\bigcup_{\widetilde{W}\in E_1}\!\powerset (\widetilde{W})$ and $\sigma\in \Sigma$, $W\stackrel{\sigma}{\longrightarrow}W'$ in $\mathcal{A} (E_1)$ iff $W\stackrel{\sigma}{\longrightarrow}W'$ in $\mathcal{A} (E_2)$.\vspace{5pt}
		\end{proof}
	\end{lemma}

		The following lemma shows that the operator $\mathcal{A}(\cdot)$ introduced in Definition \ref{def-supervisor SE} provides a way to construct a cofinal subset of $\SPR(G,R)$. That is, for any $S\in \SPR(G,R)$, we have $S||G\sqsubseteq \mathcal{A} (E)||G$ for some $\Sigma_{ucr}$-controllability set $E$.

	\begin{lemma}\label{lemma-S smaller than AES}
		If $S\in \SPR(G,R)$, then $S||G\sqsubseteq \mathcal{A} (E)||G$ for some $\Sigma_{ucr}$-controllability set $E$ from $G$ to $R$. In particular, $S||G\sqsubseteq \mathcal{A} (E (\Phi))||G$ for any $\Phi:S||G\sqsubseteq_{cc} R$.\\

		\begin{proof}
			Let  $\Phi:S||G\sqsubseteq_{cc} R$. By Lemma \ref{lemma-S implies exists ES}, $E (\Phi)$  (see, Def.\ \ref{def-EPhi}) is a $\Sigma_{ucr}$-controllability set from $G$ to $R$. By Definitions \ref{def-supervisor SE} and \ref{def-EPhi}, we may assume $\mathcal{A} (E (\Phi))= (\bigcup_{y\in Y}\powerset  (\theta_y),\Sigma,\longrightarrow, I_{E (\Phi)})$ with\vspace{-5pt}	
			$$ I_{E (\Phi)}=\{W_0\in \bigcup_{y\in Y}\powerset  (\theta_y): \forall x_0\in X_0 \, \exists z_0\in Z_0\ $$  $$( (x_0,z_0)\in W_0) \text{ and }W_0\subseteq X_0\times Z_0\}.$$
			Here, $\theta_y=\{ (x,z): ( (y,x),z)\in \Phi \text{ and }$ $(y,x)$ is reachable in $S||G$\} for each $y\in Y$.
			It  suffices to show that the relation $\Psi\subseteq  (Y\times X)\times (\bigcup_{y\in Y}\powerset  (\theta_y)\times X)$  defined as  
			\begin{align*}\nonumber
				\Psi=\{& ( (y,x), (W,x)): W\mathbin\subseteq \theta_y\text{ and }
				\exists s\mathbin\in \Sigma^* 	( (y,x) \text{ and }  (W,x) \\ 
				& 				
				\text{ are $s$-reachable in } S||G \text{ and } \mathcal{A} (E (\Phi))||G \text{ resp.})   \}\nonumber
			\end{align*}
			
	\noindent		is a simulation relation from $S||G$ to $\mathcal{A} (E(\Phi))||G$. 
			
			First, it will be verified that $\Psi$ satisfies the condition  (initial state) in Definition \ref{def-simulation}.
			Let $ (y_0,x_0)\in Y_0\times X_0$. For any $x_0'\in X_0$, since $\Phi$ is a cc-simulation relation from $S||G$ to $R$ and $ (y_0,x_0')\in Y_0\times X_0$, there exists $z_0'\in Z_0$ such that $ ( (y_0,x_0'),z_0')\in \Phi$, and hence we can choose arbitrarily and fix such a state  $z_0'$ and denote it as $\Delta (x_0')$. Set  $W_0=\{ (x_0',\Delta (x_0')): x_0'\in X_0\}.$
			Then $W_0\subseteq X_0\times Z_0$ and 	
			$$ \forall x_0'\in X_0 \, \exists z_0'\in Z_0 \  ( (x_0',z_0')\in W_0) \text{ and }  $$ $$ \forall  (x_0',z_0')\in W_0 \, ( ( (y_0,x_0'),z_0')\in \Phi).$$ 
			Hence,  $W_0\subseteq \theta_{y_0}$, $W_0\in \bigcup_{y\in Y}\powerset  (\theta_y)$ and $W_0\in I_{E (\Phi)}$, which implies $ ( (y_0,x_0), (W_0,x_0))\in \Psi$.	
			
			Next, we deal with the condition  ($\Sigma$-forward) in Definition \ref{def-simulation}.
			Let $ ( (y,x), (W,x))\in \Psi$, $ (y,x)\stackrel{\sigma}{\longrightarrow} (y^*,x^*)$ and $\sigma\in \Sigma$. Thus, $W\subseteq \theta_y$ and both $ (y,x)$ and $ (W,x)$ are $s$-reachable for some $s\in \Sigma^*$. Hence,  $ (y_0,x_{01})\stackrel{s}{\longrightarrow} (y,x)$  and $ (W_0,x_{02})\stackrel{s}{\longrightarrow} (W,x)$ for some $y_0\in Y_0$, $x_{01}\in X_0$, $W_0\in I_{E (\Phi)}$ and $x_{02}\in X_0$.  We will show that $W\stackrel{\sigma}{\longrightarrow}W^*$ and $ ( (y^*,x^*), (W^*,x^*))\in \Psi$ for some $W^*\in  \bigcup_{y\in Y}\powerset  (\theta_y)$. To this end,  the claim below is obtained first.\\\mbox{}
			
			\textbf{Claim}  For any $ (x_1,z_1)\in W$ and $x_1'\in X$ with $x_1\stackrel{\sigma}{\longrightarrow}x_1'$, there exists $z_1'\in Z$ such that  $z_1\stackrel{\sigma}{\longrightarrow}z_1'$, $ (x_1',z_1')\in \theta_{y^*}\in E (\Phi)$ and 	$ ( (y^*,x_1'),z_1')\in \Phi$.
			
			Let $ (x_1,z_1)\in W$ and $x_1\stackrel{\sigma}{\longrightarrow}x_1'$. Then 
			$ ( (y,x_1),z_1)\in \Phi$ because of $W\subseteq \theta_y$ and $ (x_1,z_1)\in W$. By Lemma \ref{lemma-about xz in W}, it follows from $W_0\stackrel{s}{\longrightarrow}W$ and $x_1\in \pi_G (W)$ that $x_1$ is  $s$-reachable. Further, due to $y_0\stackrel{s}{\longrightarrow}y$, $ (y,x_1)$ is also  $s$-reachable in  $S||G$.  Moreover, $ (y,x_1)\stackrel{\sigma}{\longrightarrow} (y^*,x_1')$ comes from $ (y,x)\stackrel{\sigma}{\longrightarrow} (y^*,x^*)$ and $x_1\stackrel{\sigma}{\longrightarrow}x_1'$. Hence, it follows from $ ( (y,x_1),z_1)\in \Phi$  that 
			$z_1\stackrel{\sigma}{\longrightarrow}z_1'$ and 	$ ( (y^*,x_1'),z_1')\in \Phi$ for some $z_1'\in Z$. Further, since $(y^*,x_1')$ is reachable in $S||G$,	by  Definition \ref{def-EPhi}, we get $ (x_1',z_1')\in \theta_{y^*}\in E (\Phi)$, as desired. \\\mbox{}

			By the above claim,  for any $ (x_1,z_1)\in W$ and $x_1'\in X$ with $x_1\stackrel{\sigma}{\longrightarrow}x_1'$, we can choose arbitrarily and fix a state $\ast (x_1,z_1,x_1')\in Z$ such that $z_1\stackrel{\sigma}{\longrightarrow}\ast (x_1,z_1,x_1')$, $ (x_1',\ast (x_1,z_1,\\x_1'))\in \theta_{y^*}\in E (\Phi)$ and 	$ ( (y^*,x_1'),\ast (x_1,z_1,x_1'))\in \Phi$. Set\vspace{-5pt} 		 	
			$$W^*= \{ (x_1', \ast (x_1,z_1,x_1')):  (x_1,z_1)\in W \text{ and } x_1\stackrel{\sigma}{\longrightarrow}x_1'  \}.\vspace{-5pt}$$
			Clearly, $W^*\subseteq \theta_{y^*}$ and thus $W^*\in \bigcup_{y\in Y}\powerset  (\theta_y)$. Moreover, it is obvious that\vspace{-5pt}
			$$W^* \subseteq\bigcup_{ (x,z)\in W}\{x':x\stackrel{\sigma}{\longrightarrow} x'\}\times\{z':z\stackrel{\sigma}{\longrightarrow} z'\}.\vspace{-5pt}$$  
			\noindent	Hence, the clauses $ (\ref{def-supervisor SE}-b)$ and $ (\ref{def-supervisor SE}-c)$ from Definition \ref{def-supervisor SE} hold for $W$, $W^*$ and $\sigma$. In the following, we intend to show that  $(\ref{def-supervisor SE}-a)$ holds for $\sigma$, $W$ and $W^*$.
			Since $W_0\stackrel{s}{\longrightarrow}W$ and $x_{02}\stackrel{s}{\longrightarrow}x$ due to $ (W_0,x_{02})\stackrel{s}{\longrightarrow} (W,x)$, by Lemma \ref{lemma-about xz in W}, $x\in Reach (s,X_0)=\pi_G (W)$, and hence $ (x,z)\in W$ for some $z\in Z$. Further, due to $x\stackrel{\sigma}{\longrightarrow}x^*$ and the construction of $W^*$, we get $z\stackrel{\sigma}{\longrightarrow}\ast (x,z,x^*)$ and $ (x^*,\ast (x,z,x^*))\in W^*$, and thus the condition $ (\ref{def-supervisor SE}-a)$  holds for $\sigma$, $W$ and $W^*$. Thus, by Definition~\ref{def-supervisor SE}, $W\stackrel{\sigma}{\longrightarrow}W^*$, and hence	$ (W,x)\stackrel{\sigma}{\longrightarrow} (W^*,x^*)$ due to $x\stackrel{\sigma}{\longrightarrow}x^*$.
			
			We end this proof by  showing $ ( (y^*,x^*), (W^*,x^*))\in \Psi$.
			Since both $ (W,x)$ and $ (y,x)$ are $s$-reachable, 
			both $ (y^*,x^*)$ and  $ (W^*,x^*)$  are $s\sigma$-reachable because of $ (W,x)\!\stackrel{\sigma}{\longrightarrow}\! (W^*,x^*)$ and $ (y,x)\stackrel{\sigma}{\longrightarrow} (y^*,x^*)$.
			Hence,  $ ( (y^*,x^*), (W^*,x^*))\in \Psi$ due to $W^*\subseteq \theta_{y^*}$, as desired.			
		\end{proof}
	\end{lemma}

	The following theorem states that $\mathcal{A} (E_{ (G,R)}^\uparrow)$ is a maximally permissive supervisor  (w.r.t.\ the simulation relation) for the $ (G,R)$-similarity control problem with required events.
	
	\begin{theorem}\label{th-SEarrow is a supervisor}
		For any plant  $G$ and  specification $R$, if $\SPR\\ (G,\!R)\mathbin{\neq}\emptyset$, then $\mathcal{A} (E_{ (G,R)}^\uparrow)\mathbin\in \MSPR (G,\!R)$.\\
		
		\begin{proof}
			By Proposition \ref{prop-fixpoint}, Corollary \ref{corollary- fixpoint is control set} and Theorem \ref{th-uc admissbile is sufficient}, one has $\mathcal{A} (E_{ (G,R)}^\uparrow)\mathbin\in \SPR (G,R)$. Moreover, by Lemma  \ref{lemma-AE smaller than fixpoint} and \ref{lemma-S smaller than AES}, it is easy to see that \plat{$S||G\sqsubseteq \mathcal{A} (E_{ (G,R)}^\uparrow)||G$} for any $S\in \SPR (G,R)$, as desired.	
		\end{proof} 
	\end{theorem}

	\section{Relations to Other Control Problems}\label{sec:relations}
	
	In this section, we discuss the connections among three types of control problems for discrete event system. First, the  bisimilarity control problem \cite{takai2019bisimilarity, 2006Control} and  similarity control problem \cite{kushi2017synthesis} are recalled.

	\hspace{-2pt}Given a plant  $G$ and specification $R$, the $ (G,\!R)$\textit{-bisimilarity control problem}  (or, 	$ (G,R)$\textit{-similarity control problem}) refers to finding a $\Sigma_{uc}$-admissible supervisor $S$ such that $S||G\simeq R$  ($S||G\sqsubseteq R$, resp.).

	It is easy to see that $ (G,R)$-bisimilarity control problems (or,  
	$ (G,R)$-similarity control problems) can be regarded as $ (G,R)$-similarity control problems with required events with $\Sigma_r=\Sigma$  ($\Sigma_r=\emptyset$, resp.). In this sense, they are special cases of the issue considered in this paper. The notion recalled below is related to a result on the bisimilarity control problem.
	
	\begin{myDef}\label{def-state controllable}\cite{2006Control}
		Let $G$ be an automaton. A supervisor $S$ is said to be state-controllable with respect to $G$ if, for any  $\sigma\in \Sigma_{uc}$ and $s\in \Sigma^*$, 
		$$  Reach (s\sigma,X_0)\neq\emptyset \Longrightarrow \forall y\in Reach (s,Y_0)\  (y\stackrel{\sigma}{\longrightarrow}).$$
	\end{myDef}
	
	That is, the uncontrollable event $\sigma$ is enabled at all $s$-reachable states in $S$ whenever $s\sigma$ is a run in $G$.
	Necessary and sufficient conditions for the  solvability of the $ (G,R)$-bisimilarity  and  $ (G,R)$-similarity control problems have been given as below.

	\begin{theorem}\cite{kushi2017synthesis,takai2019bisimilarity,2006Control}\label{th-bi and similarity control problem solvabilty}	(1) The $ (G,R)$-bisimilarity control problem is solvable iff
		$T||G\simeq R$ for some state-controllable (w.r.t.\ $G$) automaton $T$  with state set $\powerset (X\times Z)$;

		\noindent	(2) The $ (G,\hspace{-1pt}R)$-similarity control problem is solvable iff
		$G\mathbin\sqsubseteq_{uc}   R$.

	\end{theorem}
	In \cite{kushi2017synthesis}, it is shown that, for any automaton $S$, $S$			 is state-controllable w.r.t.\ $G$ iff $S$ is $\Sigma_{uc}$-admissible w.r.t.\ $G$. We think this result weakens the significance of Theorem \ref{th-bi and similarity control problem solvabilty} (1) as a necessary and sufficient condition for the solvability of the bisimilarity control problem, because it only asserts that this problem can be solved if and only if it has a solution with $\powerset(X\times Z)$ as the state set. Of course, it isn't difficult to see that it helps to measure the complexity of verifying the existence of a supervisor in finite cases.
	As a corollary of Theorem \ref{th-uc admissbile is sufficient}, the following result provides another necessary and sufficient condition for the solvability of the bisimilarity control problem, which, similar to Theorem \ref{th-bi and similarity control problem solvabilty} (2), only involves properties of $G$ and $R$ and doesn't involve the solution itself.

	\begin{corollary}\label{corollary-E act as solutions for bisimilarity control problem}
		The $ (G,R)$-bisimilarity control problem is solvable if and only if there exists a $\Sigma_{ucr}$-controllability set $E$ from $G$ to $R$ with $\Sigma_r=\Sigma$ such that $\gamma_R (\bigcup E_0)=Z_0$, that is,  $\forall z\in Z_0 \, \exists W\in E_0 \, \exists x_0\in X_0 \, ( (x_0,z)\in W)$, where
		$$\gamma_R(\bigcup E_0)=\{z\in Z:(x,z)\in \bigcup E_0\text{ for some }x\in X\} \text{ and }$$
		$E_0=\{W\in E: \forall x\in X_0 \, \exists z\in Z_0 \, ( (x,z)\in W) \text{ and }  \\ \mbox{} \hfill W\subseteq X_0\times Z_0\}.$\\

		\begin{proof}
			($\Longleftarrow$) Consider the automaton
			$\mathcal{A} (E)= (\bigcup_{W\in E}\powerset (W),\linebreak[1] \Sigma, \longrightarrow, I_E)$ defined in Definition \ref{def-supervisor SE}. Clearly, $I_E=E_0$. As shown in the proof of Lemma \ref{lemma- E exists implies AE is a solution}, the automaton 	$\mathcal{A} (E)$ is $\Sigma_{uc}$-admissible, and the relation $\Phi_{E}\subseteq  (\bigcup_{W\in E}\powerset (W)\times X)\times Z$ defined as
			$$ \Phi_{E}=\{ ( (W,x),z):  (x,z)\in W \text{ and } W\in \bigcup_{W\in E}\powerset (W)\}$$ 
			is a binary relation from $\mathcal{A} (E)||G$ to $R$ satisfying the conditions  ($\Sigma$-forward) and  ($\Sigma_r$-backward).  In order to show that  $\Phi_{E}$ is a bisimulation from 	$\mathcal{A} (E)||G$ to $R$, by $\Sigma_r=\Sigma$, it is enough to show that
			
			\noindent$ (a)$ for any initial state $ (W_0,x_0)$ in $\mathcal{A} (E)||G$, there exists $z_0\in Z_0$ such that $ ( (W_0,x_0),z_0)\in \Phi_E$;
			
			\noindent$ (b)$ for any $z_0\in Z_0$, there exists an initial state $ (W_0,x_0)$ in $\mathcal{A} (E)||G$ such that $ ( (W_0,x_0),z_0)\in \Phi_E$.
			
			The former has been shown in the proof of Lemma \ref{lemma- E exists implies AE is a solution}; the later immediately follows from $\gamma_R (\bigcup E_0)=Z_0$. In a word, $\mathcal{A} (E)$ indeed is $\Sigma_{uc}$-admissible and $\Phi_E:\mathcal{A} (E)||G\simeq R$, as desired.

			($\Longrightarrow$) Let $S$ be a $\Sigma_{uc}$-admissible supervisor, $\Phi$ a bisimulation relation from $S||G$ to $R$ and $E (\Phi)=\{\theta_y:y\in Y \}$ (see, Definition \ref{def-EPhi}). We intend to show that $\bigcup_{W\in E(\Phi)}\powerset (W)$ is  the controllability set we are looking for.
			By Lemma \ref{lemma-S implies exists ES} and \ref{lemma-E^* is successor}, $\bigcup_{W\in E(\Phi)}\powerset (W)$ is a $\Sigma_{ucr}$-controllability set from $G$ to $R$ with $\Sigma_r=\Sigma$. It remains to prove that $\gamma_R (\bigcup E_0)=Z_0$.
			
			Clearly, $\gamma_R (\bigcup E_0)\subseteq Z_0$ because $W\subseteq X_0\times Z_0$ for each $W\in E_0$. Let $z_0\in Z_0$. Due to $\Phi: S||G\simeq R$, we get $ ( (y_0,x_0),z_0)\mathbin\in \Phi$ for some $ (y_0,x_0)\mathbin\in Y_0\times X_0$. Hence, $ (x_0,z_0)\in \theta_{y_0}$. Moreover, as shown in the proof of Lemma \ref{lemma-S implies exists ES},  $\forall x\in X_0\linebreak[4] \exists z\in Z_0 \, ( (x,z)\in \theta_{y_0})$. Furthermore, by $\theta_{y_0}\cap (X_0\times Z_0)\in \bigcup_{W\in E(\Phi)}\powerset (W)$, we have  $ (x_0,z_0)\in \theta_{y_0}\cap (X_0\times Z_0)\in E_0$. Hence, $z_0\in \gamma_R (\bigcup E_0)$. Thus, $Z_0\subseteq \gamma_R (\bigcup E_0)$, as desired.
		\end{proof}
	\end{corollary}

	In the  situation that $\Sigma_r=\emptyset$, for any binary relation $\Phi\subseteq X\times Z$, it is easy to see that $\Phi$ is a $\Sigma_{uc}$-simulation from $G$ to $R$ if and only if $\powerset (\Phi)$ is a $\Sigma_{ucr}$-controllability set. Moreover, for any $\Sigma_{ucr}$-controllability set $E$, it is obvious that $\bigcup E$ is a $\Sigma_{uc}$-simulation. 
	Thus, as a corollary of Theorem \ref{th-uc admissbile is sufficient}, clause (2) in Theorem \ref{th-bi and similarity control problem solvabilty} can also be obtained immediately.

Other necessary and sufficient conditions for the solvability of  the bisimilarity control problem are given in  \cite{8972367,takai2019bisimilarity}, where \cite{8972367}  considers it   under partial observation.\footnote{If the event set is further partitioned into an observable event set $\Sigma_o$ and an unobservable event set $\Sigma_{uo}$, i.e., $\Sigma=\Sigma_o\cup\Sigma_{uo}$ with $\Sigma_o\cap\Sigma_{uo}=\emptyset$, where the controller cannot change state in response to the execution of an unobservable event,
the system is said to be \emph{partially observed}; in case $\Sigma_o=\Sigma$ it is \emph{fully observed}.
}
 Necessarily, such three conditions must be equivalent to each other (in the case of full observation). Furthermore, based on the proof of Corollary \ref{corollary-E act as solutions for bisimilarity control problem}, Theorem 1 in \cite{8972367} and 
Theorem 26 in \cite{takai2019bisimilarity}, we can naturally derive how the three conditions can be converted into each other. 

In terms of complexity, all three verification algorithms share a similar (exponential) complexity, namely
$$O(|X|^6|Z|^6\times 2^{3|X||Z|})\mbox{\cite{8972367}}, O(2^{3|X||Z|}\times |X|^2\times|Z|^2\times|\Sigma|)\mbox{\cite{takai2019bisimilarity}}$$
$$\text{and }O (|X|^4|Z|^4\times2|\Sigma| \times 2^{2|X||Z|})\text{ respectively,}$$  
 where the complexity of the last one is given below Theorem \ref{th-uc admissbile is sufficient}, and the detailed algorithm can be found in the appendix.
All three complexities are smaller than the doubly exponential complexity of exhaustive search that uses (1) of Theorem \ref{th-bi and similarity control problem solvabilty}.

We now turn to a brief discussion on the range control problem considered in \cite{2007Control}. Given a plant $G$ and two specifications $R_1$ and $R_2$ with  $R_1\sqsubseteq  R_2$, the $(R_1,G,R_2)$-range control problem refers to finding a $\Sigma_{uc}$-compatible\footnote{A supervisor $S=(Y,\Sigma,\longrightarrow,Y_0)$ is $\Sigma_{uc}$-compatible if $y\stackrel{\sigma}{\longrightarrow}$ for each $\sigma\in\Sigma_{uc}$ and $y\in Y$.} supervisor $S$ such that $R_1\sqsubseteq S||G \sqsubseteq R_2$. Intuitively, the automata $R_1$ and $R_2$ specify the minimally adequate and maximally acceptable behavior of the supervised system $S||G$ respectively. In other words, $S||G$ must implement the behavior specified by $R_1$ while all its behavior must be permitted by $R_2$. From this, it is not difficult to see the intuitive similarities between this problem and the one considered in this paper. However, there  are significant differences between them at technical levels. 
First, for range control problems, minimally adequate behavior of the supervised system is characterized in terms of automata, while this paper emphasizes that it must respect the occurrence of certain events (i.e.\ events in $\Sigma_r$) when implementing specifications.
Secondly, the inequality $R_1\sqsubseteq S||G \sqsubseteq R_2$ involves two simulation relations, while $S||G\sqsubseteq_{cc}R$ involves only one cc-simulation relation. The difference between these two requirements has significant consequences. For example, in the circumstance where $R$ has only one initial state, the $(G,R)$-bisimilarity control problem can be  reduced trivially to the $(G,R)$-similarity control problem with $\Sigma_r=\Sigma$, while it seems far from trivial to find specifications $R_1$ and $R_2$ for a given plant $G$ and specification $R$ such that the $(G,R)$-bisimilarity control problem and the $(R_1,G,R_2)$-range control problem have exactly the same solution set.
In summary, we believe that the problem considered in this paper can not be reduced to the range control problem, although they share similar motivations.  We conclude this section by proposing two connections between them. Let $R_i \ (i=1,2)$ be two automata and $\Sigma'\subseteq \Sigma$. A binary relation $\Phi$ from $R_1$ to $R_2$ is said to be a simulation w.r.t.\ $\Sigma'$ if $\Phi$ satisfies the conditions (initial state) and (forward) with $\sigma\in\Sigma'$ in Definition \ref{def-simulation}.

\begin{proposition}\label{prop-simulation and range control 1}
  Let $G$ be a plant, $R_i=(Z_i,\Sigma,\longrightarrow,Z_{0i}) \ (\mbox{for}\\ i=1,2)$ specifications with $R_1\sqsubseteq R_2$ and $S$ a $\Sigma_{uc}$-compatible solution of the $(R_1,G,R_2)$-range control problem. If there exist simulations $\Phi: R_1\sqsubseteq R_2$ and $\Phi_1:R_1\sqsubseteq S||G$ such that $\Phi^{-1}$ is a simulation from $R_2$ to $R_1$ w.r.t.\ $\Sigma'$ and $\Phi_1^{-1}\circ\Phi: S||G\sqsubseteq R_2$%
\,\footnote{$\Phi\circ\Psi = \{(x,z)\mid \exists y ~ ((x,y)\in \Phi \wedge (y,z) \in \Psi)\}$ denotes the relational composition of binary relations $\Phi$ and $\Psi$ \cite{Howie95}.}, then $S$ is a $\Sigma_{uc}$-admissible supervisor such that $S||G\sqsubseteq_{cc} R_2$ with $\Sigma_r=\Sigma'$, that is, $S\in \SPR(G,R_2)$, where $\Sigma'=\{\sigma\in\Sigma: z_1\stackrel{\sigma}{\longrightarrow} \text{ in } R_1 \text{ for some } z_1\in Z_1\}$. \\
	
	\begin{proof}
		By Definition \ref{def-admissible uc}, it's obvious that $S$ is a $\Sigma_{uc}$-admissible supervisor. To complete the proof, it's sufficient to show that $\Phi_1^{-1}\circ\Phi$ is a cc-simulation relation from $S||G$ to $R_2$.
		Due to $\Phi_1^{-1}\circ\Phi: S||G\sqsubseteq R_2$, it  remains to verify  the ($\Sigma_r$-backward) condition. Let $((y,x),z_2)\in \Phi_1^{-1}\circ\Phi$ and $z_2\stackrel{\sigma}{\longrightarrow}z_2'$ with $\sigma\in \Sigma_r$. Then $((y,x),z_1)\in \Phi_1^{-1}$ and $(z_1,z_2)\in \Phi$ for some $z_1\in Z_1$, and hence $(z_2,z_1)\in \Phi^{-1}$. Since  $\Phi^{-1}$ is a simulation from $R_2$ to $R_1$ w.r.t.\ $\Sigma'$ and $\sigma\in \Sigma_r$ (=$\Sigma'$), it follows that $z_1\stackrel{\sigma}{\longrightarrow}z_1'$ with $(z_2',z_1')\in \Phi^{-1}$ for some $z_1'\in Z_1$. By $(z_1,(y,x))\in \Phi_1$ and $\Phi_1: R_1\sqsubseteq S||G$, we get $(y,x)\stackrel{\sigma}{\longrightarrow}(y',x')$  with $((y',x'),z_1')\in \Phi_1^{-1}$ for some $(y',x')\in Y\times X$, and thus $((y',x'),z_2')\in \Phi_1^{-1}\circ\Phi$, as desired.
	\end{proof}
\end{proposition}

We also have the following proposition, where $S(R_1,\!G,\!R_2)$ is the set of all $\Sigma_{uc}$-admissible solutions to the $(R_1,G,R_2)$-range control problem.
	
	\begin{proposition}
		Let $G$, $R_1$, $R_2$ and $\Sigma'$ be as in the above proposition, and let $S\in \SPR(G,R_2)$ with $\Sigma_r=\Sigma'$. If there exists a cc-simulation $\Phi$ from $S||G$ to $R_2$ such that $\Phi^{-1}$ is a simulation
		from $R_2$ to $S||G$ w.r.t.\ $\Sigma'$, then $S\in S(R_1,G,R_2)$. In particular, $\SPR(G,R_2)\subseteq S(R_1,G,R_2)$ whenever $|Z_{02}|\mathbin=1$.\\
		
		\begin{proof}
			Clearly, there exists a simulation relation $\Phi'$ from $R_1$ to $R_2$. Since $\Phi^{-1}$ is a simulation from $R_2$ to $S||G$ w.r.t.\ $\Sigma'$ $(\subseteq \Sigma)$, $\Phi'\circ\Phi^{-1}$ is a simulation from $R_1$ to $S||G$ w.r.t.\ $\Sigma'$. Further, since $\Sigma'$ contains all enabled events in $R_1$, $\Phi'\circ\Phi^{-1}$ is a simulation from $R_1$ to $S||G$, and hence $R_1\sqsubseteq S||G\sqsubseteq R_2$.

			Let $S\in \SPR(G,R_2)$. Thus, $S||G\sqsubseteq_{cc} R_2$.
			In the situation that $|Z_{02}|=1$, for each $\Psi:S||G\sqsubseteq_{cc}R_2$,		
			$\Psi^{-1}$ satisfies the condition (initial state) and (forward) with $\sigma\in\Sigma'$, and hence $\Psi^{-1}$ is a simulation from $R_2$ to $S||G$ w.r.t.\ $\Sigma'$. Therefore $S\in S(R_1,G,R_2)$, as desired.\vspace{2ex}
		\end{proof}
	\end{proposition}
	
	\section{Conclusion}\label{sec:conclusion}

	This paper considers the similarity control problem with required events based on the notion of covariant-contravariant simulation, which is an extension of both the bisimilarity control and similarity control problems explored in the literature. In this situation, the behavior of the supervised system not only needs to be allowed by the specification, but also needs to respect  requirements prescribed by the specification saying that in certain states certain events must be enabled.
	
	A necessary and sufficient condition for the existence of solutions to this problem is obtained in terms of $\Sigma_{ucr}$-con\-trollability sets. Moreover, a method for synthesizing a maximally permissive supervisor is provided. However,  the issue of nonblockingness (see,  \cite{2019Maximally, 2018Nonblockingsimilarity}) is not considered in this paper; this is an interesting topic for future work.
	
	\section*{Acknowledgments}

The authors would like to thank the Associate Editors and the anonymous reviewers for their valuable comments and constructive suggestions, which have significantly improved the quality of this manuscript.
The research of the third author was supported by Royal Society Wolfson Fellowship RSWF\textbackslash R1\textbackslash 221008.

    \section*{References}
        \vspace{-2em}
    \bibliographystyle{plain}
	\bibliography{Reference}

  \appendix
        \section{Appendix: Algorithms}

This appendix solves the similarity control problem with required events and synthesizes a maximally permissive supervisor by an algorithmic approach. By
Theorem \ref{th-uc admissbile is sufficient}, the solvability of this problem can be determined by checking the existence of a $\Sigma_{ucr}$-controllability set.

Given two automata $G$ and $R$, we first present the algorithm for checking the predicate $match_{G,R}(W, \sigma, W')$.

\begin{algorithm}
\caption{Verifying $match_{G,R}(W, \sigma, W')$}
\label{alg:match-verification}
\begin{algorithmic}[1]
\Require  $W $ and $W'$, and  $\sigma \in \Sigma$
\Ensure Return true if $match_{G,R}(W, \sigma, W')$, otherwise false
\State \textbf{boolean function } $match_{G,R}(W, \sigma, W')$
\ForAll{$(x, z) \in W$}
\ForAll{$x' \in X$ with $x \stackrel{\sigma}{\longrightarrow} x'$}
\State $\mathit{valid} \gets \text{true}$
\If{$\,\not\!\exists$ $z'$ with $z\stackrel{\sigma}{\longrightarrow}z'$ and $(x', z') \in W'$}
\State $\mathit{valid} \gets \text{false}$
\State \textbf{break}
\EndIf
\EndFor
\EndFor
\State \Return \text{true}
\end{algorithmic}
\end{algorithm}

\begin{algorithm}
\caption{Checking Existence of a $\Sigma_{ucr}$-Controllable Set}
\label{algorithm-fourth-control-set-existence}
\begin{algorithmic}[1]
\Require Automata $G$ and $R$
\Ensure Return true if there exists a $\Sigma_{ucr}$-controllable set $E$, otherwise false

\State $C \gets \powerset(X \times Z)$
\State $\mathit{changed} \gets \text{true}$

\While{$\mathit{changed}$}
\State $\mathit{changed} \gets \text{false}$
\ForAll{$W \in C$} \Comment{Iterate over elements $W$ in $C$} 
\ForAll{$\sigma \in \Sigma_{uc}$}\Comment{Check condition $(\ref{def-controlability}-a)$}
\If{$\exists W' \mathbin\in C$ such that $match_{G,R}(W,\sigma,W')$} 
\State $\mathit{found} \gets \text{true}$
\Else
\State $C \gets C - \{W\}$
\State $\mathit{changed} \gets \text{true}$
\State \textbf{break} \Comment{Proceed to next $W$}
\EndIf
\EndFor

\If{$W \in C$}\\
\Comment{If $W$ has not been removed, check condition $(\ref{def-controlability}-b)$}
\ForAll{$(x, z) \in W$ and $\sigma \in \Sigma_r$}
\ForAll{$z' \in Z$ such that $z \stackrel{\sigma}{\longrightarrow} z'$}
\State $\mathit{found} \gets \text{false}$
\ForAll{$W' \in C$}
\If{$match_{G, R}(W, \sigma, W')$ $\&\&$ $\exists x' \in X$ such that $x \stackrel{\sigma}{\longrightarrow} x'$ and $(x', z') \in W'$}
\State $\mathit{found} \gets \text{true}$
\State \textbf{break} 
\EndIf
\EndFor
\If{not $\mathit{found}$}
\State $C \gets C - {W}$
\State $\mathit{changed} \gets \text{true}$
\State \textbf{break} \Comment{Exit current loop}
\EndIf
\EndFor
\EndFor
\If{$\mathit{changed}$} \textbf{break} \EndIf \Comment{If $W$ is removed, exit the $(x,z)$ loop}
\EndIf
\EndFor
\EndWhile
\ForAll{$W_0 \in C$}\Comment{Check the condition (i-state)}
\If{$\forall x_0 \in X_0, \exists z_0 \in Z_0$ such that $(x_0, z_0) \in W_0$}
\State \Return \text{true}
\EndIf
\EndFor
\State \Return \text{false}
\end{algorithmic}
\end{algorithm}

Algorithm \ref{algorithm-fourth-control-set-existence} adopts an iterative approach, starting from the set of all possible state pairs (i.e., $\powerset(X \times Z)$) and progressively removing elements that do not satisfy the controllability conditions until a fixed point is reached. Then it verifies whether the fixed point contains an element satisfying condition (i-state) of Definition \ref{def-controlability}.

Given $W,W'\in E\subseteq \powerset (X\times Z)$ and $\sigma\in \Sigma$, the time complexity of verifying $match_{G,R}(W,\sigma,W')$ in Algorithm \ref{alg:match-verification} is $O(|X|^2|Z|^2)$. Therefore, for each $W\in E$, the time complexities of checking whether $W$ satisfies clause $(\ref{def-controlability}-a)$ (or, $(\ref{def-controlability}-b)$) is 
$O(|X|^2|Z|^2\times|\Sigma_{uc}|\times|E|) $ ($O(|X|^2|Z|^2\times|\Sigma_r|\times|E|\times |X|^2|Z|^2)$, resp.)

Consequently, the time complexity of checking whether $W$ satisfies both clauses $(\ref{def-controlability}-a)$ and $(\ref{def-controlability}-b)$ is $O(|X|^4|Z|^4\times|(\Sigma_{uc}+\Sigma_r)|\times|E|)$. Since $E$ can have at most $2^{|X||Z|}$ elements, the time complexity of deciding the existence of a solution to the $(G,R)$-simulation control problem with requirement events is at most $$O (|X|^4|Z|^4\times(|\Sigma|+|\Sigma_r|) \times 2^{2|X||Z|}).$$ 

Algorithm \ref{algorithm-construct S(E)} constructs the supervisor automaton $\mathcal{A}(E)$ based on a given $\Sigma_{ucr}$-controllability set $E$. 
It  builds the state set in the first step and then computes the initial state set of $\mathcal{A}(E)$.  Finally it verifies the transition conditions $(\ref{def-supervisor SE}-a,b,c)$ of Definition \ref{def-supervisor SE} by invoking Algorithm \ref{algorithm-check-transition-conditions of SE}. \begin{algorithm} 
\caption{Constructing Automaton $\mathcal{A}(E)$} \label{algorithm-construct S(E)} 
\begin{algorithmic}[1] 
\Require Automata $G$ and $R $, and the $\Sigma_{ucr}$-controllability set $E$ 
\Ensure Automaton $\mathcal{A}(E) = (Q, \Sigma, \rightarrow, I_E)$ 
\State $Q \gets \bigcup_{\widetilde{W}\in E}\powerset (\widetilde{W})$ 
\State $I_E \gets \emptyset$ \ForAll{$W_0 \in Q$} 
\If{$W_0 \subseteq X_0 \times Z_0$ \textbf{and} $\forall x_0 \in X_0, \exists z_0 \in Z_0$ such that $(x_0, z_0) \in W_0$} 
\State $I_E \gets I_E \cup \{W_0\}$ \EndIf \EndFor 
\ForAll{$W \in Q$} \ForAll{$\sigma \in \Sigma$ and $W' \in Q$} \If{$checkconditions(W,\sigma,W')$} \State add $W \stackrel{\sigma}{\longrightarrow} W'$ to $\mathcal{A}(E)$ 
\EndIf 
\EndFor
\EndFor 
\State 
\Return $\mathcal{A}(E) = (Q, \Sigma, \rightarrow, I_E)$ 
\end{algorithmic} 
\end{algorithm} 
\begin{algorithm} \caption{Checking $(\ref{def-supervisor SE}-a,b,c)$ of Definition \ref{def-supervisor SE}}
\label{algorithm-check-transition-conditions of SE}
\begin{algorithmic}[1] 
\Require $G, R, W, W' \subseteq X \times Z$, $\sigma$ 
\Ensure Return true if conditions $(\ref{def-supervisor SE}-a,b,c)$ of Definition \ref{def-supervisor SE} holds, otherwise false 
\State \textbf{boolean function } $checkconditions(W,\sigma,W')$ 
\State $\mathit{condA} \gets\text{false}$ 
\ForAll{$(x,z) \in W$} 
\ForAll{$(x',z') \in W'$} 
\If{$x \stackrel{\sigma}{\longrightarrow} x'$  and $z \stackrel{\sigma}{\longrightarrow} z'$ } 
\State $\mathit{condA} \gets\text{true}$ 
\State \textbf{break} 
\EndIf
\EndFor 
\If{$\mathit{condA}$} \textbf{break} \EndIf \EndFor 
\State $\mathit{condB} \gets match_{G, R}(W, \sigma, W')$ 
\State $\mathit{condC} \gets \text{true}$ 
\ForAll{$(x',z') \in W'$} 
\State $\mathit{found} \gets \text{false}$ 
\ForAll{$(x,z) \in W$} \If{$x \stackrel{\sigma}{\longrightarrow} x'$ $\&\&$ $z \stackrel{\sigma}{\longrightarrow} z'$} 
\State $\mathit{found} \gets \text{true}$ 
\State \textbf{break}
\EndIf 
\EndFor
\If{not $\mathit{found}$} \State $\mathit{condC} \gets \text{false}$ \State \textbf{break} 
\EndIf 
\EndFor
\If{$\mathit{condA}$ $\&\&$ $\mathit{condB}$ $\&\&$ $\mathit{condC}$} 
\State \Return \text{true}
\Else 
\State \Return \text{false} 
\EndIf 
\end{algorithmic}
\end{algorithm}

$\!$Algorithm \ref{alg:greatet fixpoint} generates the greatest fixpoint $E_{G,R}^{\uparrow}$ of  $F_{(G,R)}$. Subsequently, by calling Algorithm \ref{algorithm-construct S(E)}, we obtain the automaton $\mathcal{A} (E_{ (G,R)}^\uparrow)\mathbin\in \MSPR (G,\!R)$, which is a maximally permissive supervisor due to Theorem \ref{th-SEarrow is a supervisor}.

\begin{algorithm}
\caption{Constructing the greatest Fixpoint of $F_{(G,R)}$}
\label{alg:greatet fixpoint}
\begin{algorithmic}[1]
\Require Automata $G$ and  $R$
\Ensure The greatest fixpoint $E_{G,R}^{\uparrow}$ of $F_{(G,R)}$
\State $E \gets \powerset(\powerset(X\times Z))$ \Comment{Initialize}
\Repeat
    \State $E_{\text{prev}} \gets E$
    \State $E_{\text{new}} \gets \emptyset$
    \ForAll{$W \in E_{\text{prev}}$}
    
        \State $cond1 \gets \text{true}$
        \ForAll{$\sigma \in \Sigma_{uc}$}\Comment{Check (1) of Definition \ref{def-function}}
            \State $\text{exists\_match} \gets \text{false}$
            \ForAll{$W' \in E_{\text{prev}}$}
                \If{$match_{G,R}(W,\sigma,W')$}
                    \State $\text{exists\_match} \gets \text{true}$
                    \State \textbf{break}
                \EndIf
            \EndFor
            \If{$\neg \text{exists\_match}$}
                \State $cond1 \gets \text{false}$
                \State \textbf{break}
            \EndIf
        \EndFor
        \If{$\neg cond1$}
            \State \textbf{continue} \Comment{Skip $W$ if (1) fails}
        \EndIf

        \State $cond2 \gets \text{true}$ \Comment{Check (2) of Definition \ref{def-function}}
        \ForAll{$(x, z) \in W$ and $\sigma \in \Sigma_r$}
            
                \ForAll{$z' \in Z$ such that $z \xrightarrow{\sigma}_R z'$} 
                        \State $\text{found} \gets \text{false}$
                        \ForAll{$x' \in X$ such that $x \xrightarrow{\sigma} x'$}
                            
                                \ForAll{$W' \in E_{\text{prev}}$}
                                    \If{$(x', z') \in W'$ \&\&\newline\hfill\mbox{}\hspace{1.4in} $match_{G,R}(W,\sigma,W')$}
                                        \State $\text{found} \gets \text{true}$
                                        \State \textbf{break}
                                    \EndIf
                                \EndFor
                            
                            \If{$\text{found}$}
                                \State \textbf{break}
                            \EndIf
                        \EndFor
                        \If{$\neg \text{found}$}
                            \State $cond2 \gets \text{false}$
                            \State \textbf{break}
                        \EndIf

                \If{$\neg cond2$}
                    \State \textbf{break}
                \EndIf
            \EndFor
 
        \EndFor
        
        \If{$cond1$  \&\& $cond2$}
            \State $E_{\text{new}} \gets E_{\text{new}} \cup \{W\}$
        \EndIf
    \EndFor
    \State $E \gets E_{\text{new}}$
\Until{$E = E_{\text{prev}}$} \Comment{No change between iterations}
\State \Return $E$
\end{algorithmic}
\end{algorithm}

	\small

	\vspace{-2ex}
	\newpage
	\begin{wrapfigure}[6]{l}{1in}
		{\includegraphics[width=1in,height=1.25in,clip,keepaspectratio]{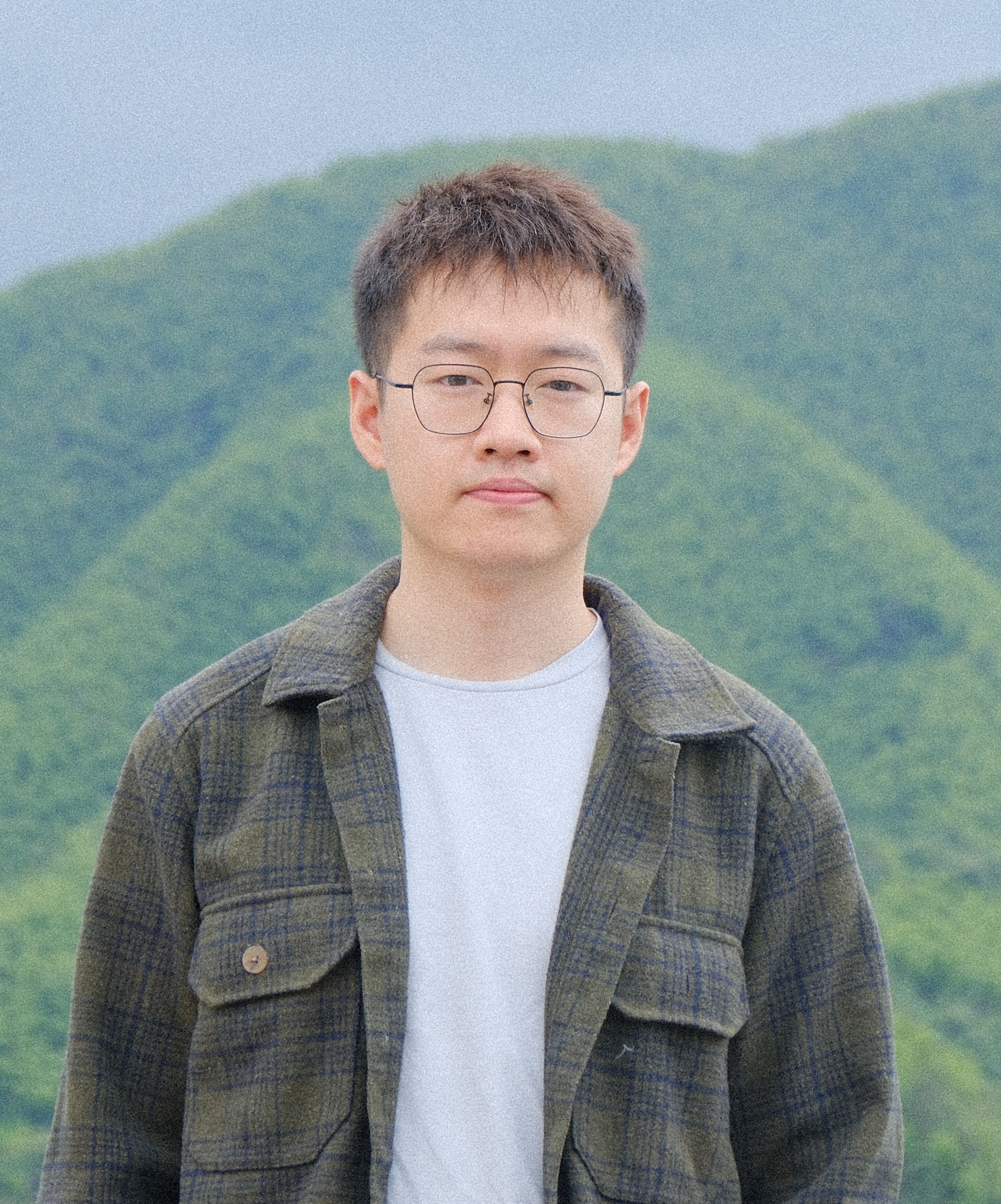}}
	\end{wrapfigure}
	\mbox{}\\[3ex]\textbf{Yu Wang} is currently a Ph.D. student  at the  College of Computer Science and Technology, Nanjing University of Aeronautics and Astronautics. His research includes formal methods and logic in computer science.
	\vspace{3ex}
	
	\begin{wrapfigure}[9]{l}{1in}
		{\includegraphics[width=1in,height=1.25in,clip,keepaspectratio]{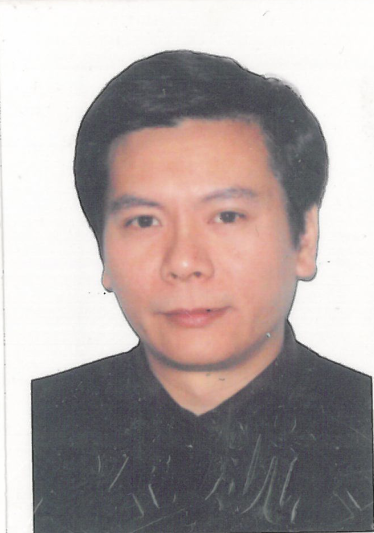}}
	\end{wrapfigure}
	\mbox{}\\[3ex]\textbf{Zhaohui Zhu} received the B.E., M.E. and Ph.D degree in computer science and technology from Nanjing University of Aeronautics and Astronautics (NUAA), Nanjing, China, in 1992, 1995 and 1998, respectively. From 1998 to 2000, he was a postdoctor at Nanjing University. From 2000 to 2003, he was an associate professor at NUAA. Since 2004, he has been a professor at NUAA. His research interests include formal methods and applied logic in computer science and AI.
\vspace{1ex}
	\vfill
	
	\begin{wrapfigure}{l}{1in}
		{\includegraphics[width=1in,height=1.25in,clip,keepaspectratio]{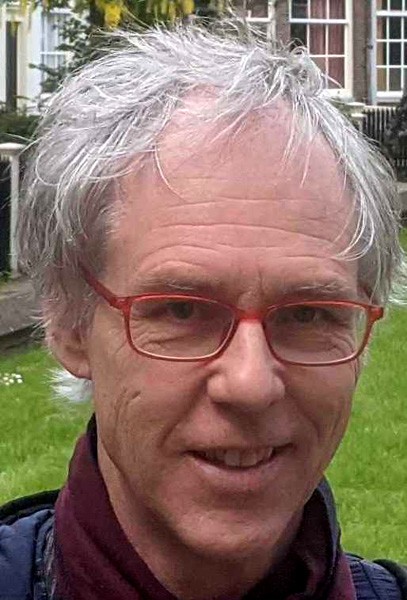}}
	\end{wrapfigure}
	\noindent\textbf{Rob van Glabbeek} (MAE) received his PhD degree in computer science from the Free University of Amsterdam in 1990. He spend 12 years at Stanford University and 18 years in Sydney, Australia -- from 2016 to 2022 as Chief Research Scientist at Data61, CSIRO. Since 2022 he is a Royal Society Wolfson Fellow at the University of Edinburgh. He won the 2020 CONCUR test-of-time award, is a foreign member of the Royal Holland Society of Sciences and Humanities,
	fellow of the Asia-Pacific Artificial Intelligence Association, and member of Academia Europaea.
	
	Prof.\ van Glabbeek is known for the conciliation of the interleaving and true concurrency communities by codeveloping the current view of branching time and causality as orthogonal but interacting dimensions of concurrency. Together with Peter Weijland he invented the notion of branching bisimulation, that since has become the prototypical example of a branching time equivalence, and the semantic equivalence most used in most verification tools. With Dominic Hughes he made a crucial contribution to the proof theory of linear logic by proposing a notion of proof net that had been sought after in vain by linear logicians since the inception of linear logic. Together with Vaughan Pratt he initiated the now widespread use of higher dimensional automata and other geometric models of concurrency. With Peter H\"ofner, he introduced \emph{justness}, an assumption in between progress and weak fairness, that is crucial for proving liveness properties of distributed system.
	
	He is editor-in-chief of Electronic Proceedings in Theoretical Computer Science, a member of the editorial boards of Information and Computation and Theoretical Computer Science, and has been on dozens of program committees.
	\vspace{1ex}

	\begin{wrapfigure}[9]{l}{1in}
		{\includegraphics[width=1in,height=1.25in,clip,keepaspectratio]{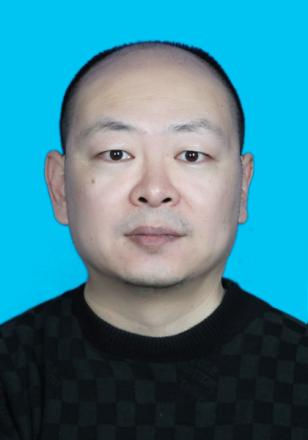}}
	\end{wrapfigure}
	\mbox{}\\[1ex]\textbf{Jinjin Zhang} received his Ph.D. degree in computer science and technology from the Institute of Computing Technology, Nanjing University of Aeronautics and Astronautics, Nanjing, in 2011. He is an associate professor of Nanjing Audit University, Nanjing. His current research interests include software engineering and logic in computer science.
	\vspace{4ex}
	
	\begin{wrapfigure}{l}{1in}
		{\includegraphics[width=1in,height=1.25in,clip,keepaspectratio]{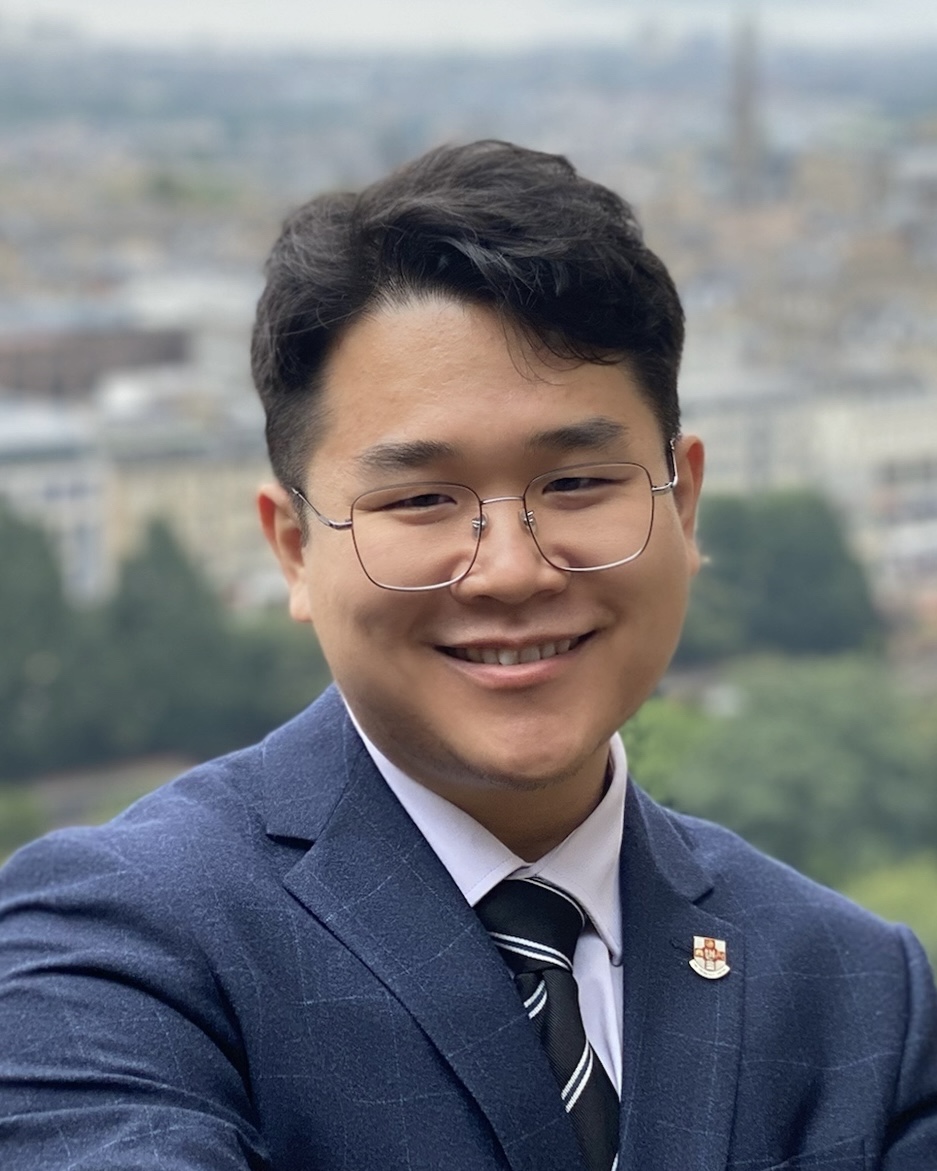}}
	\end{wrapfigure}
	\noindent \textbf{Yixuan Li} received his Ph.D. from the University of Edinburgh. His research bridges LLMs and formal methods to develop user-friendly frameworks for automated reasoning.
\end{document}